\documentclass[journal]{IEEEtran}

\usepackage{amsmath,amsfonts}
\usepackage{amsthm}
\usepackage{algorithmic}
\usepackage{algorithm}
\usepackage{array}
\usepackage[caption=false,font=normalsize,labelfont=sf,textfont=sf]{subfig}
\usepackage{textcomp}
\usepackage{stfloats}
\usepackage{url}
\usepackage{verbatim}
\usepackage{graphicx}
\usepackage{cite}
\usepackage{soul}
\usepackage{xcolor}
\usepackage{hhline} 
\usepackage{siunitx}

\usepackage{tikz}
\usetikzlibrary{positioning}
\usepackage{pgfplots}

\pgfplotsset{compat=newest}
\usetikzlibrary{shapes,arrows,fit,positioning,calc}
\usetikzlibrary{plotmarks}
\usetikzlibrary{decorations.pathreplacing}
\usetikzlibrary{patterns}
\usetikzlibrary{chains,arrows,shapes,spy}
\usetikzlibrary{automata}

\usepackage{textcomp}
\usetikzlibrary{shapes,arrows,backgrounds}
\usetikzlibrary{trees}

\usepackage{enumitem}

\newtheorem{proposition}{Proposition}
\newtheorem{theorem}{Theorem}
\newtheorem{lemma}{Lemma}

\begin{document}

\title{Channel State Information Preprocessing for CSI-based Physical-Layer Authentication Using Reconciliation} 


\author{Atsu Kokuvi Angélo Passah, Rodrigo C. de Lamare, ~\IEEEmembership{Fellow,~IEEE,} and Arsenia Chorti, ~\IEEEmembership{Senior Member,~IEEE,}
\thanks{Atsu Kokuvi Angélo Passah is with ETIS Laboratory, ENSEA, CY Cergy Paris University, CNRS, France, and with the Pontifical Catholic University of Rio de Janeiro (PUC-Rio), Brazil.}%
\thanks{Arsenia Chorti is with ETIS Laboratory, ENSEA, CY Cergy Paris University, CNRS, France, and with Barkhausen Institut gGmbH, Germany.}%
\thanks{Rodrigo C. de Lamare is with the Pontifical Catholic University of Rio de Janeiro (PUC-Rio), Brazil, and with the School of Physics, Engineering and Technology, York University, United Kingdom.}
}



\maketitle

{
\begin{abstract}
This paper introduces an adaptive preprocessing technique to enhance the accuracy of channel state information-based physical layer authentication (CSI-PLA) alleviating CSI variations and inconsistencies in the time domain. To this end, we develop an adaptive robust principal component analysis (A-RPCA) preprocessing method based on robust principal component analysis (RPCA). 
The performance evaluation is then conducted using a PLA framework based on information reconciliation,  in which Gaussian approximation (GA) for Polar codes is leveraged for the design of short codelength Slepian Wolf decoders. Furthermore, an analysis of the proposed A-RPCA methods is carried out.
Simulation results show that compared to a baseline scheme without preprocessing and {without reconciliation}, the proposed A-RPCA method substantially reduces the error probability after reconciliation and also substantially increases the detection probabilities that is also $1$ in both line-of-sight (LOS) and non-line-of-sight (NLOS) scenarios. We have compared against state-of the-art preprocessing schemes in both synthetic and real datasets, including principal component analysis (PCA) and robust PCA, autoencoders and the recursive projected compressive sensing (ReProCS) framework and we have validated the superior performance of the proposed approach.
\end{abstract}
}


\section{Introduction}
\IEEEPARstart{T}{he} advent of next generation wireless communication systems is anticipated to substantially enhance wireless connectivity \cite{Aazhang}. Unprecedented capabilities, massive device density, terahertz communications and the seamless integration of artificial intelligence (AI) and machine learning (ML) at the network edge are envisioned. These advances are essential to support emerging paradigms including the deployment of autonomous agents, immersive extended reality and large scale Internet of things (IoT). 

However, in many forthcoming sixth-generation (6G) verticals, deploying standard cryptographic authentication handshakes is challenging for several reasons, including extremely stringent latency requirements, massive-scale low-end IoT deployments, and strict constraints on memory, energy, and computational resources. As a result, IoT systems and networks of cyber–physical agents are increasingly vulnerable to novel and sophisticated threats that cannot be fully mitigated by conventional cryptographic methods, particularly in power- and computation-constrained devices and in machine-to-machine communications requiring ultra-low latency \cite{Shakiba}.

Moreover, attacks occurring during network entry or link establishment are especially difficult to counter, as they take place \textit{before} any authentication protocol can be executed. Prominent examples come under the generic term ``false base station (FBS) attack'' \cite{Karaçay}, in which adversaries deploy counterfeit base stations to deceive users into setting up links with them , and, eventually, intercepting sensitive information. Consequently, complementary authentication mechanisms are required.

In this context, physical layer security (PLS), and in particular physical layer authentication (PLA), have emerged as promising complementary approaches. PLA leverages the intrinsic properties of the wireless channel or the inherent hardware imperfections of radio-frequency (RF) devices to distinguish legitimate transmitters from malicious ones. Techniques based on channel state information (CSI) or RF fingerprinting have been extensively studied as viable solutions for fast and lightweight authentication \cite{Mitev}.

In more detail, CSI-based PLA identifies a device or user by comparing observed channel state information (CSI) with pre-recorded reference CSI. It typically follows a two-phase procedure comprising an enrollment phase and an authentication phase. During the offline enrollment phase, the authenticator collects and stores CSI vectors associated with the legitimate user, which serve as a baseline reference. In the subsequent online authentication phase, newly observed CSI vectors are acquired and compared against the stored baseline to verify the transmitter’s identity. However, raw CSI measurements are inherently sensitive to distortions caused by noise, hardware impairments, and user mobility.
As a result, depending on the propagation environment and scenario—such as line-of-sight (LoS) or non-line-of-sight (NLoS) conditions—it can be challenging to maintain robust authentication performance. To mitigate these inconsistencies in CSI estimation, appropriate preprocessing techniques are therefore required.

Albeit, the following caveat must be taken into consideration. CSI-based PLA fundamentally relies on the correlation between channel measurements across authentication phases, where CSI observations are expected to be more strongly correlated under normal conditions than under adversarial ones. Even when two CSI measurements exhibit high correlation in certain feature dimensions, independently applying preprocessing to each phase can inadvertently destroy this correlation—particularly when the correlated components lie in the noisiest dimensions. Consequently, although independent preprocessing may improve denoising or filtering, it does not necessarily preserve, and may even degrade, the correlation structure that is critical for reliable authentication, ultimately reducing PLA performance.

To illustrate this effect, consider principal component analysis (PCA) preprocessing in two scenarios: (i) when the correlation between raw CSI at times 
$t$
(enrollment) and 
$t+1$ (authentication) is concentrated in the dominant principal components, and (ii) when the correlation is not aligned with the principal components. Fig.~\ref{corr12} illustrates both cases. In scenario (i) (Fig.~\ref{corr1}), the correlation remains visible before and after PCA reconstruction, increasing from 0.70 to 0.89 after preprocessing. In contrast, in scenario (ii) (Fig.~\ref{corr2}), the correlation decreases significantly after PCA reconstruction, dropping from 0.49 to 0.10, as the correlated components are discarded while retaining only the principal components.

This example demonstrates that independently applying preprocessing across authentication phases can degrade performance under normal conditions rather than improve it. Therefore, adaptive preprocessing schemes that explicitly account for channel measurements from previous authentication phases are necessary to preserve correlation and ensure reliable CSI-based PLA, which constitutes the key motivation behind the proposed approach in this manuscript.
\begin{figure}[!t]
  \centering
  \subfloat[\footnotesize Correlations concentrated in the largest principal components \label{corr1}]{
    \includegraphics[width=0.7\linewidth]{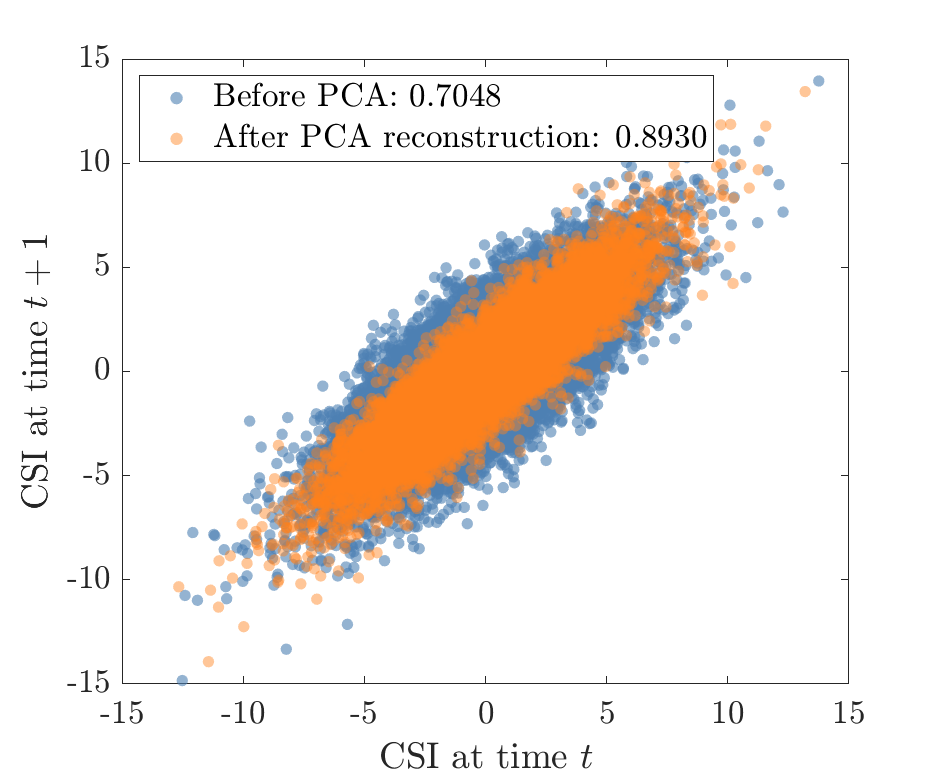}  
  } \hfill 
  \subfloat[\footnotesize Correlations not concentrated in the largest principal components \label{corr2}]{
    \includegraphics[width=0.7\linewidth]{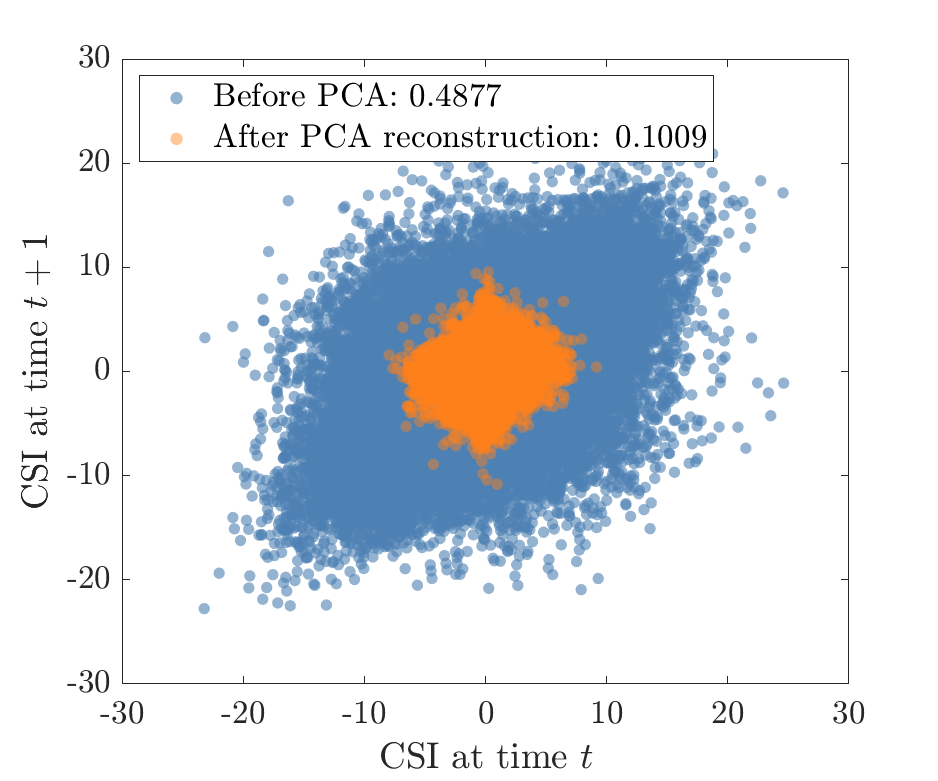}
    }
  \caption{Representation of the relationship between CSI dataset at time $t$ and $t+1$ before PCA and after PCA reconstruction: (i) when the correlation is concentrated in top principal components, the preprocessing can improve the correlation and (ii) when the correlation is not concentrated in the largest principal components, the correlation cannot be improved.}
  \label{corr12}
\end{figure}

\subsection{Previous and Related Works}

In recent years, CSI-PLA has attracted significant attention as a lightweight alternative to conventional authentication methods. In this subsection, we first review low-rank signal processing techniques that will serve to underpin the proposed approach, followed by a more general overview of recent related works on CSI-based PLA.

\subsubsection{Low-Rank Signal Processing and Robust Subspace Learning}

Principal component analysis (PCA) has been among the earliest preprocessing techniques investigated for CSI-based PLA. In \cite{Srinivasan,Srinivasan2}, PCA was introduced to reduce dimensionality and mitigate noise effects in CSI measurements, with experimental validation conducted in outdoor wireless environments.

Beyond PCA, more advanced robust subspace learning methods have been extensively studied. A comprehensive overview of such techniques was presented in \cite{Vaswani}, covering robust principal component analysis (RPCA) based on low-rank plus sparse decomposition, as well as robust subspace tracking (RST) and robust subspace recovery (RSR) methods. While RST focuses on tracking a slowly evolving low-dimensional subspace in the presence of outliers, RSR assumes that entire data vectors are either inliers or outliers, in contrast to RPCA, which models outliers as sparse corruptions superimposed on a low-rank structure.

The RPCA problem was formally introduced in \cite{Candès} through the principal component pursuit (PCP) formulation, a convex optimization framework that minimizes a weighted combination of the nuclear norm and the $\ell_1$-norm. This approach was shown to successfully recover both low-rank and sparse components even under severe corruption or missing data. Extensions of RPCA have since been proposed; for instance, \cite{Bian} incorporated feature selection by employing a $\sigma$-norm for reconstruction error and an $\ell_{2,0}$-norm constraint for subspace projection.

To address dynamic scenarios, the recursive projected compressive sensing (ReProCS) framework was proposed as an online solution for robust subspace tracking and dynamic RPCA when the underlying subspace evolves slowly over time \cite{6875221,6638807,6854385,7032141,8437747,8736886}. Theoretical guarantees on exact sparse recovery and bounded reconstruction errors were established in \cite{6638807,6875221}. Practical implementations, referred to as Prac-ReProCS, were introduced in \cite{6854385}, with extensive experimental validation reported in \cite{7032141}.

While classical ReProCS assumes slowly varying subspaces, the work in \cite{8736886} proposed an online moving window RPCA (OMWRPCA) algorithm capable of handling both gradual and abrupt subspace changes via change-point detection. Furthermore, a simplified variant known as simple-ReProCS was introduced in \cite{8437747}, relaxing several assumptions of earlier methods and achieving improved outlier tolerance and reduced memory complexity.

Low-rank signal processing techniques have also been explored in other contexts. The work in \cite{jidf} proposed an adaptive low-rank framework for interference suppression based on joint iterative least squares optimization, along with low-complexity LMS and RLS algorithms. In \cite{sor}, the subspace-orbit randomized singular value decomposition (SOR-SVD) was introduced to efficiently compute accurate low-rank matrix approximations. Additionally, \cite{8478175} formulated a nonconvex, nonsmooth sparse and low-rank decomposition problem solved via the alternating direction method of multipliers (ADMM), employing weighted nuclear norms for improved rank approximation. Building upon these developments, in this work we propose an adaptive RPCA-based preprocessing technique for PLA using synthetic data in \cite{passah2024wcnc}.

\subsubsection{Physical Layer Authentication}

Several PLA schemes have been proposed based on different channel features. In \cite{7420748}, an authentication framework exploiting the channel impulse response (CIR) was developed, incorporating multipath delay characteristics through a two-dimensional quantization procedure to mitigate random amplitude and delay fluctuations. To avoid performance degradation caused by quantization errors, the authors of \cite{9335644} proposed CIR-based PLA schemes without quantization. Alternatively, \cite{10056964} investigated a PLA scheme based on the channel phase response, mitigating performance loss due to channel correlation; closed-form expressions for the mean, variance, and probability density functions were also derived.

More recently, we introduced reconciliation-based PLA schemes that leverage information reconciliation techniques to align channel measurements across time \cite{ref_vtc2024,info_reco2}. These approaches employ quantization combined with error-correcting codes, such as polar codes \cite{ArikanPolar,arikan2009,trifonov2012,pga}, to enhance detection performance. While \cite{ref_vtc2024} focused on a single-user scenario, \cite{info_reco2} extended the framework to a multi-user setting with interfering users. In both cases, reconciliation-based PLA was shown to outperform existing schemes. 

Overall, existing works highlight the importance of robust preprocessing and subspace-aware techniques for reliable CSI-based PLA, particularly in dynamic and noisy wireless environments, thereby motivating the adaptive low-rank preprocessing approach proposed in this paper.

\subsection{Contributions}

This work presents a PLA framework based on adaptive preprocessing techniques and information reconciliation using polar codes to enhance the accuracy of CSI-based PLA and alleviate CSI estimation variations and inconsistencies in the time domain. 
In particular, an
adaptive RPCA (A-RPCA) preprocessing is proposed.
An analysis of the proposed 
A-RPCA methods is carried out along with a study of their computational cost. Unlike the works in \cite{ref_vtc2024} and \cite{info_reco2}, in this work we employ Gaussian approximation (GA) for designing polar codes instead of the code construction that relies on the binary erasure channel (BEC). The main contributions of this paper can be summarized as follows. 
\begin{itemize}
    \item A PLA framework based on an adaptive preprocessing technique denoted by A-RPCA is proposed to enhance the accuracy of CSI-PLA and alleviate CSI estimation uncertainties and time-varying nature.
    \item The proposed adaptive preprocessing takes into account both the CSI from the enrollment and the authentication phases instead of applying the preprocessing independently across the phases. This is achieved through the correlation coefficient between the CSI from the enrollment and authentication phases and it allows to discriminate the impact of the adversary CSI in contrast to the legitimate user's CSI. 
    \item A convergence analysis of the time-regularized principal component pursuit (TR-PCP) optimization problem that results in the proposed A-RPCA algorithm is carried out. 
    \item We assess the preprocessing and the PLA authentication performance with synthetic data and in real-world scenarios, where a dataset from Nokia is used to distinguish between different users. CSI-PLA using information reconciliation is applied to a real dataset unlike in \cite{info_reco2} \item  We have compared against state-of the-art preprocessing schemes in both synthetic and real datasets, including principal component analysis (PCA) and robust PCA, autoencoders and the recursive projected compressive sensing (ReProCS) framework and we have validated the superior performance of the proposed approach.
\end{itemize}

{
The remainder of the paper is organized as follows. Section \ref{sys_model} includes the system and channel models. The proposed 
A-RPCA technique along with the optimization problem used to formulate it and the algorithm convergence analysis are presented in Section 
\ref{proposed_prepro}. Then in Section \ref{auth_setup}, the PLA setup is described, including the information reconciliation and the polar code construction 
and numerical results are outlined in Section \ref{results} where the performance using synthetic data and real dataset are evaluated. Conclusions are carried out in \ref{conclude}.
}


\section{System Model} 
\label{sys_model}


In our system model, three wireless communication nodes, referred to as Bob, Alice $1$ and Alice $2$\footnote{Although it is typical to refer to a node of interest as Alice and any another node as``Eve'', in this work we use alternatively the terms Alice $1$ and Alice $2$ because i) we do not study optimal attack strategies from an adversary, and ii) we want to highlight the separability between users and attainable accuracy with CSI pre-processing.} are considered. Bob acts as the authenticator and Alice $1$ is the user who Bob wants to identify based on the CSI observed at two different time slots $t$ and $t+1$, 
in the presence of Alice $2$. 
$t$ and $t+1$ correspond respectively to the enrollment phase and the authentication phase (validation) of the authentication process. Alice $1$ and Alice $2$ are single-antenna users and Bob is an $N_b$-antenna base station. 
The enrollment and authentication phases are described in the following along with the corresponding channel estimation models. 
\subsubsection{Enrollment phase (time $t)$}
During the time slot $t$, the CSI of Alice $1$
estimated at Bob is given by 
\begin{equation}
\hat{h}_{1,i}(t) = h_{1,i}(t) + z_{1,i}(t), \, i = 1, \ldots N_b,
\end{equation}
where $z_{1,i}(t) \sim \mathcal{CN}(0,1)$ is a complex Gaussian distributed noise with zero mean and unit variance. 
We consider a compound channel model in this work in equation (\ref{eqRice}) that is a Rician fading model as the envelope or magnitude $|h_{1,i}(t)|$ is Rician distributed \cite{8809413, 9716123}, so that,
\begin{equation}
        h_{1,i}(t) = \nu + \; \sigma\, x_{1,i}(t), \, i = 1, \ldots N_b. 
    \label{eqRice}
\end{equation}
In Eq. (\ref{eqRice}), $\nu = \sqrt{\frac{K_r}{K_r+1}}$ is the LoS component and $\sigma\, x_{1,i}(t)$ is the random multipath NLoS component where $K_r$ is the Rician $K$-factor and and $\sigma = \sqrt{\frac{1}{K_r+1}}$. We assume a narrowband system in which the multipath components are unresolvable and the aggregate contribution of a large number of scattered components can be modeled as a complex, circularly symmetric Gaussian random variable $x_{1,i}(t) \sim \mathcal{CN}(0,1)$. Thus $\sigma\, x_{1,i}(t)$ is a random multipath component.
When $K_r = 0$ the model reduces to a non-line-of-sight (NLoS) Rayleigh channel and when $K_r \rightarrow \infty$ it reduces to a pure LoS channel. The Rician model holds when $0 < K_r < \infty$.

\subsubsection{Authentication phase (time $t+1$)}
During the time slot $t+1$, the CSI of Alice $1$ or Alice $2$ observed by Bob can be expressed as 
\begin{equation}
\hat{h}_{u,i}(t+1) = h_{u,i}(t+1) + z_{u,i}(t+1), \, i = 1, \ldots N_b,
\end{equation}
where $z_{u,i}(t+1) \sim \mathcal{CN}(0,1)$ and $u \in \{1, 2\}$, $1$ and $2$ denote respectively Alice $1$ and Alice $2$.
$h_{u,i}(t+1)$ is assumed to follow a Rician channel model (\ref{eqRice2}) as previously
\begin{equation}
        h_{u,i}(t+1) = \nu + \; \sigma\, x_{u,i}(t+1), \, i = 1, \ldots N_b.
    \label{eqRice2}
\end{equation}
We assume a Markov chain for time domain and spatial dependencies, captured in the correlation of $x_{1,i}(t+1)$ and $x_{1,i}(t)$  as follows \cite{info_reco2}: 
\begin{equation}
x_{1,i}(t+1) = \beta \; x_{1,i}(t) + \sqrt{1-\beta^2}\; w(t+1),
\end{equation}
where $\beta$ is the correlation coefficient and $w(t+1) \sim \mathcal{CN}(0,1)$ is a measurement noise.
After concatenating the real and imaginary parts of $\hat{h}_{1,i}(t)$ and $\hat{h}_{u,i}(t+1)$, we obtain respectively the channel matrices $\widehat{\mathbf{H}}_1(t)$ and $\widehat{\mathbf{H}}_u(t+1)$ $\in  \mathbb{R}^{N_s \times N_b}$, where $N_s$ is the number of samples.

Fig. \ref{auth} describes the PLA authentication scheme including the following steps: CSI estimation, CSI prepropcessing, and reonciliation. The reconciliation includes the quantization of the preprocessed CSI and then the Slepian-Wolf decoder using polar codes. The quantized CSI at time $t$ is encoded and then a side information is used at $t+1$ to decode the CSI at $t+1$. The side information contains the frozen bits positions and the cyclic redundancy check ($\mathrm{CRC}$) bits. The output reconciled vectors are them compared to make authentication decision. This scheme will be described in full detail in Section IV, while in the next section we will focus on the proposed A-RPCA approach.

\begin{figure} 
\centering
\includegraphics[width=0.5\textwidth]{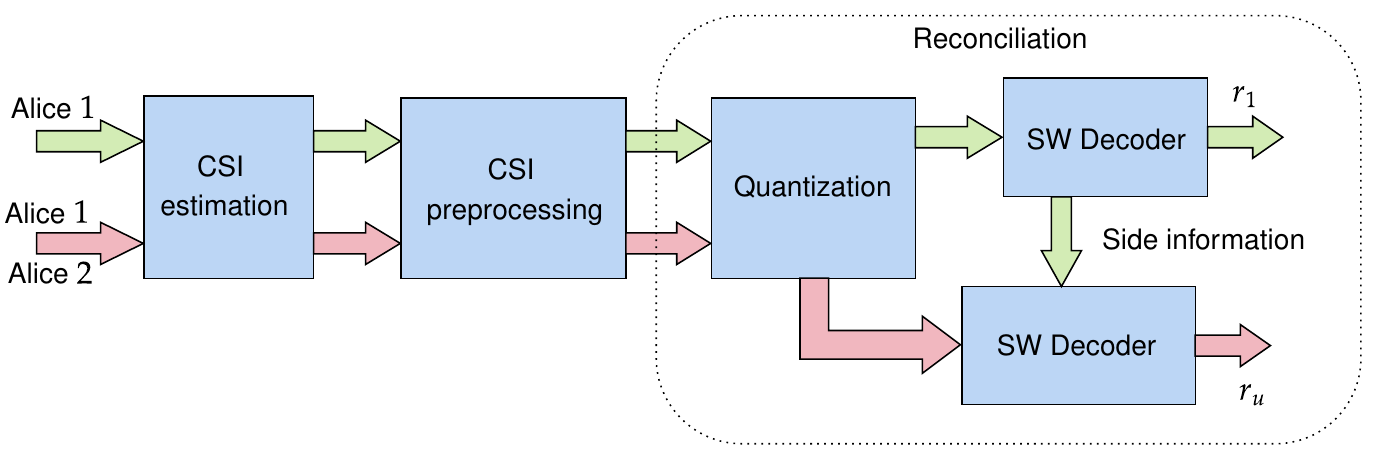} 
\vspace{-1em}
\caption{
Block diagram of the adaptive preprocessing scheme for CSI-PLA. The green arrow is the path of CSI of Alice  $1$ at time $t$ and the pink one is the path of the CSI of Alice $1$ or Alice $2$ during the next time slot $t+1$. The preprocessed channel state information are introduced in the input of the reconciliation block where a quantization is then applied followed by the Slepian-Wolf decoding and the outputs are $\mathbf{r}_1$ and $\mathbf{r}_u$, $u \in \{1, 2\}$.}   
\label{auth}
\end{figure}

\section{Proposed Adaptive RPCA Preprocessing} \label{proposed_prepro}


In this section, we present the proposed A-RPCA preprocessing scheme, including its phases, the formulation of the TR-PCP optimization problem that results in the proposed A-RPCA algorithm and its convergence analysis. First, let us provide in subsection \ref{sb_prev} some pertinent algorithmic details from the state-of-the-art approaches that will be used in comparsison and then outline the proposed A-RPCA in subsection \ref{sb_arpca}.

\subsection{Previous works} \label{sb_prev}

It should be noted that in these schemes, the preprocessing is applied independently  across authentication phases. The description is shown here at time $t$ and it is similarly applied at $t+1$.

\subsubsection{PCA preprocessing} 
The PCA preprocessing \cite{Srinivasan} applied at time $t$ to $\widehat{\mathbf{H}}_1(t)$ is described as follows.
\begin{itemize}
    \item Calculate the covariance matrix of $\widehat{\mathbf{H}}_1(t)$: $\mathbf{C}_H (t)$ 
    \item Compute the eigenvalue decomposition of $\mathbf{C}_H (t)$ and denote $\mathbf{U}_H \in \mathbb{C}^{N_b \times N_b}$ the matrix such that the rows are the eigenvectors of $\mathbf{C}_H$ sorted in decreasing order of the eigenvalues
    \item The channel data $\widehat{\mathbf{H}}_1(t)$ is then projected on the $D$ principal components axis as
     \(   \mathbf{W}_H = \mathbf{U}_{H, \, 1:D}\widehat{\mathbf{H}}_1(t) \)
    \item The PCA reconstruction that corresponds to the predictable part of the observed CSI, which is suitable for PLA, is obtained as follows:
\(    
\widetilde{\mathbf{H}}_1(t) = \mathbf{U}_{H, \,1:D}^\top \mathbf{W}_H \).
\end{itemize}

\subsubsection{RPCA preprocessing}

The RPCA method formulated via the PCP algorithm \cite{Candès} is applied to decompose the CSI of Alice $\widehat{\mathbf{H}}_1(t)$ into two components, namely, the low-rank matrix $\mathbf{L}_1(t)$ and the sparse matrix $\mathbf{S}_1(t)$. The low-rank component corresponds to the effective CSI needed for authentication, and the sparse one represents the outliers or noise that is removed from the raw CSI. In equation (\ref{rpca}), the PCP optimization problem consists of minimizing of the weighted sum of the nuclear norm and the $\ell_1$-norm. The nuclear norm $\|\mathbf{L}_1(t)\|_{*}$ enforces the low-rank structure, the $\ell_1$-norm enforces sparsity, and $\lambda$ is a regularization parameter used to determine the trade-off between the effective CSI and the outlier. 

\begin{equation}\label{rpca}
\begin{aligned}
\min_{\mathbf{L}_1(t),\,\mathbf{S}_1(t)} ~ &\|\mathbf{L}_1(t)\|_{*} + \lambda\|\mathbf{S}_1(t)\|_{1}, ~~\\
& \text{s.t.} ~~~\mathbf{L}_1(t) + \mathbf{S}_1(t) = \widehat{\mathbf{H}}_1(t),
\end{aligned}
\end{equation}

\subsubsection{Autoencoder (AE) preprocessing}
AE is an unsupervised machine learning scheme that involves an encoder and a corresponding decoder. During the encoding step, the AE maps the input estimated CSI into $D_{AE}$ dimensional encoded values. Then, during the decoder step, the AE preprocessing maps the encoded values back to the original space, which is a CSI reconstruction. During the learning process, both steps are optimized to minimize a cost function, typically the mean square error (MSE) between the input and the output \cite{Srinivasan}. $\widehat{\mathbf{H}}_1(t)$ and $\widehat{\mathbf{H}}_u(t+1)$ are respectively used as input of the AE (input of the encoder) at time $t$ and $t+1$. In the comparisons with AE based preprocessing in the results section, we will provide details regarding the number of neural network layers, the numberof neurons per layer, the activation functions, etc. 

\subsubsection{ReProCS preprocessing} The ReProCS preprocessing method \cite{Vaswani} was also adapted to our proposed PLA framework as follows. We considered that the channel measurement from $t$ to $t+1$ lied in a slowly changing low-dimensional subspace and denoted by $\mathbf{m}_1(t) \in \mathbb{R}^{N_b}$ a vector corresponding to a row of the estimated CSI $\widehat{\mathbf{H}}_1(t)$ at time $t$ with
\(
    \mathbf{m}_1(t) = \mathbf{\ell}_1(t) + \mathbf{s}_1(t),
\)
where $\mathbf{\ell}_1(t)$ is the true data and $\mathbf{s}_1(t)$ the outlier. The objective was to track $\mathbf{\ell}_1(t)$ and the span\footnote{Span refers to the set of all linear combinations of the columns of the matrix that defines the subspace} of the subspace at time $t$ by using an estimated initial subspace. The initial subspace was computed using the training data, i.e., the first few samples  $\mathbf{\ell}_1(t)$ from $\widehat{\mathbf{H}}_1(t)$. Then, during the next time $t+1$, the same training data was applied to track $\mathbf{\ell}_u(t+1)$ and the span of the subspace. Note that 
\(
    \mathbf{m}_u(t) = \mathbf{\ell}_u(t) + \mathbf{s}_u(t),
\)
and $\mathbf{m}_u(t+1) \in \mathbb{R}^{N_b}$ is a vector corresponding a row of the estimated CSI $\widehat{\mathbf{H}}_u(t+1)$.
\vspace{-0.3cm}
\subsection{Proposed Adaptive RPCA} \label{sb_arpca}
\begin{figure}[t]
\centering
\includegraphics[width=0.36\textwidth]{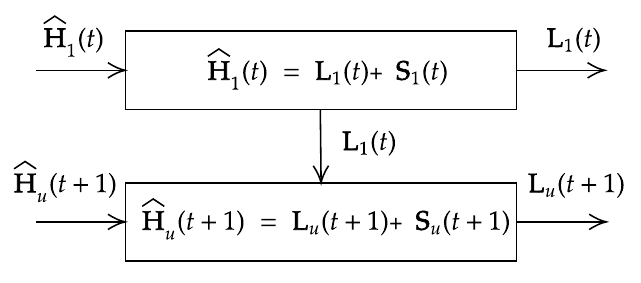}  
\caption{A-RPCA preprocessing. The effective CSI (low-rank matrix) at time $t$ is used as helper during next time $t+1$ in the temporal consistency term to perform A-RPCA.
}   
\label{arpca_scheme}
\end{figure}
As mentioned in \cite{passah2024wcnc}, we describe here the preprocessing scheme based on the formulation of the TR-PCP optimization problem, which results in the proposed A-RCPA algorithm. A-RPCA is represented in Fig. \ref{arpca_scheme} and is described in the following.
Unlike the previous methods, we exploit the temporal correlation by introducing a term in the optimization that tracks the correlation between the effective CSI at time $t$ and the effective CSI at $t+1$. 

\subsubsection{Enrollment phase}
During the offline phase, the standard RPCA method described previously is applied to output $\mathbf{L}_1(t)$ and $\mathbf{S}_1(t)$. The low-rank matrix $\mathbf{L}_1(t)$ will be used during next authentication phase.



\subsubsection{Authentication phase}

First, let us note that  $\mathbf{L}_1(t)$ is recorded at time $t$. During this phase, Bob applies the A-RPCA method to the new recorded CSI measurements $\widehat{\mathbf{H}}_u(t+1)$, taking into account the previous recorded CSI $\mathbf{L}_1(t)$. In the proposed A-RPCA algorithm, a temporal consistency term $\gamma_u \bigl\|\mathbf{L}_u(t+1) - \beta_u\mathbf{L}_1(t)\bigr\|_{F}^{2}$ is introduced in the objective function of the optimization. Therefore, the proposed TR-PCP optimization problem that results in the proposed A-RPCA algorithm is described by 
\begin{equation}\label{arpca}
\begin{aligned}
 \min_{\mathbf{L}_u(t+1), \, \mathbf{S}_u(t+1)} & \|\mathbf{L}_u(t+1)\|_{*} + \lambda\|\mathbf{S}_u(t+1)\|_{1} \; +\; \\
& \gamma_u \bigl\|\mathbf{L}_u(t+1)  - \beta_u \mathbf{L}_1(t)\bigr\|_{F}^{2} \\ 
& \text{s.t.} ~~ \mathbf{L}_u(t+1) + \mathbf{S}_u(t+1) = \widehat{\mathbf{H}}_u(t+1).
\end{aligned}
\end{equation}
where the new CSI at time $t+1$ is decomposed into the low-rank component $\mathbf{L}_u(t+1)$ and the sparse part $\mathbf{S}_u(t+1)$ with $\widehat{\mathbf{H}}_u(t+1) =  \mathbf{L}_u(t+1) + \mathbf{S}_u(t+1)$. 

In (\ref{arpca}), the temporal consistency term represents the square of the Frobenius norm ($\|\cdot\|_F$) between the current matrix $\mathbf{L}_u(t+1)$ and the previous matrix $\mathbf{L}_1(t)$. $\beta_u$, $u \in \{1, 2\}$, is the estimated correlation coefficient between $\widehat{\mathbf{H}}_1(t)$ and $\widehat{\mathbf{H}}_u(t+1)$ and $\gamma_u$ is a regularization parameter used to control the temporal dependency between $\mathbf{L}_u(t+1)$ and $\mathbf{L}_1(t)$ by determining how much importance is given to obtain a smooth evolution between them. When we are in the normal case, i.e., Alice, the impact of the CSI at $t+1$ is well considered in contrast to the other case where it is discriminated. 

In (\ref{arpca}), the nuclear norm, the $\ell_1$-norm, the square of the Frobenius norm and the constraint are convex. Consequently, the TR-PCP optimization problem is convex. As done for PCP, an augmented Lagrangian multiplier (ALM) \cite{AugmentedL} approach using ADMM \cite{Sparse} is adopted in order to split the problem into two subproblems; one that serves to find the optimal value of $\mathbf{L}_u(t+1)$ and the other that serves to find the optimal value of $\mathbf{S}_u(t+1)$. In fact, ADMM applied to low-rank signal processing \cite{jidf,sor} allows one to solve convex optimization problems by breaking them into smaller problems. Therefore, the augmented Lagrangian is first minimized with respect to $\mathbf{L}_u(t+1)$, then with respect to $\mathbf{S}_u(t+1)$, followed by an update of the Lagrange multiplier. We adopt the following notation for simplicity when solving the problem in (\ref{arpca}): $\mathbf{L}_u := \mathbf{L}_u(t+1)$, $\mathbf{S}_u := \mathbf{S}_u(t+1)$, $\widehat{\mathbf{H}}_u := \widehat{\mathbf{H}}_u(t+1)$ and $\mathbf{L}_1 =: \mathbf{L}_1(t)$.
Therefore, the optimization problem becomes
\begin{equation}\label{tr-pcp}
\begin{aligned}
\min_{\mathbf{L}_u, \, \mathbf{S}_u} \; & \|\mathbf{L}_u\|_{*} + \lambda\|\mathbf{S}_u\|_{1} + \gamma_u \,\bigl\|\mathbf{L}_u - \beta_u \,\mathbf{L}_1\bigr\|_{F}^{2}, \;\; \\
&\text{s.t.} \;\; \mathbf{L}_u + \mathbf{S}_u = \widehat{\mathbf{H}}_u,
\end{aligned}
\end{equation}
The augmented Lagrangian of (\ref{tr-pcp}) is given by 
\begin{equation} \label{auL}
\begin{aligned}
\mathcal{L}_\mu(\mathbf{L}_u~,\mathbf{S}_u~,\mathbf{Y})
= ~& \|\mathbf{L}_u\|_{*} + \lambda\|\mathbf{S}_u\|_1 + 
 \gamma_u \|\mathbf{L}_u-\beta_u \mathbf{L}_1\|_F^2 ~+\\
 & \hspace{-0.4cm} \langle \mathbf{Y} ~,~ \widehat{\mathbf{H}}_u - \mathbf{L}_u - \mathbf{S}_u \rangle + \tfrac{\mu}{2} \|\widehat{\mathbf{H}}_u - \mathbf{L}_u - \mathbf{S}_u\|_F^2,
\end{aligned}
\end{equation}
where $\mathbf{Y}$ is the Lagrangian multiplier, $\mu>0$ is the penalty term and $\langle \cdot \rangle$ is the trace inner product. 

Let us now describe the ADMM splitting as follows: 
\begin{enumerate}[label=(\alph*)]
\item $\mathbf{L}_u$-update\\ 
We fix $\mathbf{S}_u=\mathbf{S}_u^{k}$, $\mathbf{Y}=\mathbf{Y}^k$ and solve equation (\ref{eq_L}) where $\mathcal{L}_{\mu_k}$ is minimized with respect to $\mathbf{L}_u$ as follows:
\begin{equation}
\mathbf{L}_u^{k+1} = \arg\min_{\mathbf{L}_u} \mathcal{L}_{\mu_k}(\mathbf{L}_u,\mathbf{S}_u^k,\mathbf{Y}^k).
\label{eq_L}
\end{equation}
The closed-form solution of the optimal value of $\mathbf{L}_u^{k+1}$ is given in Proposition \ref{l_updtae}. 

\begin{proposition}\label{l_updtae}  
The closed-form solution of (\ref{eq_L}) is given by
\begin{equation}
\mathbf{L}_u^{k+1} =  \mathcal{D}_{\tau}(\mathbf{T}^k) = \mathbf{U}\,\Bigl[\Sigma - \tau\,\mathbf{I}\Bigr]_{+}\,\mathbf{V}^{\top},\\
\end{equation}
with,  
  $ \mathbf{T}^k = \dfrac{\mu_k\mathbf{B}^k + 2\gamma_u\beta_u\mathbf{L}_1}{\mu_k + 2\gamma_u}$ and
   $ \mathbf{B}^k = \widehat{\mathbf{H}}_u - \mathbf{S}_u^k + \tfrac{1}{\mu_k}\mathbf{Y}^k$.
 $\mathcal{D}_{\tau}(\mathbf{T}^k)$ is the singular-value thresholding, that is, $\mathbf{T}^k=\mathbf{U}\mathbf{\Sigma} \mathbf{V}^\top$, then $[\mathbf{\Sigma}-\tau \mathbf{I}]_+$ shrinks singular values. $\mathbf{I}$ is the identity matrix and \((x)_+ = \max\{x,0\}\) is applied element-wise.
 \end{proposition}
\begin{proof}
See Appendix \ref{appA}
\end{proof}

\item $\mathbf{S}_u$-update\\
 For the \(\mathbf{S}_u\)-update, we fix \(\mathbf{L}_u = \mathbf{L}_u^{k+1}\), \(\mathbf{Y} = \mathbf{Y}^{k}\) and then solve equation (\ref{eq_S}) where $\mathcal{L}_{{\mu}_k}$ is minimized with respect to $\mathbf{S}_u$, which yields
\begin{equation} \label{eq_S}
\mathbf{S}_u^{k+1} = \arg\min_{\mathbf{S}_u} \; \mathcal{L}_{\mu}(\mathbf{L}_u^{k+1},\mathbf{S}_u,\mathbf{Y}^k)
\end{equation}
The closed-form solution of the optimal value of $\mathbf{S}_u^{k+1}$ is given in Proposition \ref{s_updtae}. 

\begin{proposition}\label{s_updtae}  
The closed-form solution of (\ref{eq_S}) is given by
\begin{equation}
\mathbf{S}_u^{k+1} = \mathcal{S}_{\lambda/\mu_{k}}\bigl(\mathbf{C}^{\,k}\bigr), 
\end{equation}
with
$\mathbf{C}^{\,k}
= \widehat{\mathbf{H}}_u - \mathbf{L}_u^{k+1} + \frac{1}{\mu_{k}}\,\mathbf{Y}^{k}$
and $\mathcal{S}_{\lambda/\mu_{k}}(\cdot)$ is the soft‐threshold operator applied entrywise to $\mathbf{C}^{\,k}$,
\begin{equation}
\mathcal{S}_{\lambda/\mu_{k}}(c)
= \mathrm{sign}(c)\,\max\Bigl(\lvert c\rvert - \frac{\lambda}{\mu_{k}},\,0\Bigr).
\end{equation}
\end{proposition}

\begin{proof}
See Appendix \ref{appB}
\end{proof}

\item Once $\mathbf{L}_u^{k+1}$ and $\mathbf{S}_u^{k+1}$ are calculated, the Lagrange multiplier $\mathbf{Y}$ is updated as 
\begin{equation}
\mathbf{Y}^{k+1} = \mathbf{Y}^k + \mu_k(\widehat{\mathbf{H}}_u - \mathbf{L}_u^{k+1} - \mathbf{S}_u^{k+1}).
\end{equation}
The penalty parameter $\mu_{k+1} = \min \{ \rho \mu_k, \mu_{\max} \}
$ with $\rho > 1$ is the amplification factor.  
\end{enumerate}

The proposed A-RPCA algorithm is summarized in algorithm \ref{algo1}. 
The penalty term is initialized as \(
      \mu_{0} = \frac{3}{\sigma_{\max}\bigl(\widehat{\mathbf{H}}_u\bigr)},
    \)
    where \(\sigma_{\max}(\widehat{\mathbf{H}}_u)\) is the largest singular value of the  matrix.  
    The constant \(3\) in the expression of \(\mu_0\) tightens the first augmented Lagrangian penalty so that the optimization constraint is enforced more firmly from the beginning. The update factor $\rho = 1.5$ accelerates convergence when the primal residual remains large, while keeping the subproblems well-conditioned, so that small input perturbations will not cause disproportionately large errors in the output.
    The regularization parameter is set as in \cite{Candès},  \(\lambda = \frac{1}{\sqrt{\max(N_s,\,N_b)}} \).
    The temporal regularization parameter \(\gamma = \beta_u\). It is chosen to be equal to the correlation coefficient $\beta_u$ so that the impact of the user Alice $2$ is discriminated while improving the ability of Bob to authenticate Alice $1$.  The stopping criterion is when the normalized primal residual (i.e.\ the relative reconstruction error) \(
    \frac{\bigl\|\widehat{\mathbf{H}} - \mathbf{L} - \mathbf{S}\bigr\|_{F}}
     {\bigl\|\widehat{\mathbf{H}}\bigr\|_{F}}\) falls below a tolerance of \(tol = 10^{-7}\).

\begin{algorithm}[!t] 
\caption{Proposed A-RPCA algorithm}\label{algo1}
\begin{algorithmic}[1]
    \STATE Initialize: $\mathbf{S}_u^0=\mathbf{0}$, $\mathbf{Y}^0=\mathbf{0}$, \, $\mu_0, \, \rho, \, \mu_{\max}, \, \lambda,\, \gamma_u, \, \beta_u$ \vspace{0.05cm}
  \STATE $\mathbf{B}^k = \widehat{\mathbf{H}}_u - \mathbf{S}_u^k + \frac{1}{\mu_k}\mathbf{Y}^k$ \vspace{0.05cm}
  \STATE $\mathbf{T}^k = \frac{\mu_k\,\mathbf{B}^k + 2\gamma_u \beta_u \mathbf{L}_1}{\mu_k+2\gamma_u}$ \vspace{0.05cm}
  \STATE $\mathbf{T}^k=\mathbf{U}\mathbf{\Sigma} \mathbf{V}^T$ \vspace{0.05cm}
  \STATE $\mathbf{L}_u^{k+1}
= \mathcal{D}_{\tau}(\mathbf{T}^k)
= \mathbf{U}\,\mathcal{S}_{\tau}(\mathbf{\Sigma})\,\mathbf{V}^\top$ \vspace{0.05cm}
  \STATE $\mathbf{C}^k = \widehat{\mathbf{H}}_u - \mathbf{L}_u^{k+1} + \frac{1}{\mu_k}\mathbf{Y}^k$ \vspace{0.05cm}
  \STATE $\mathbf{S}_u^{k+1}=\mathcal{S}_{\lambda/\mu_k}(\mathbf{C}^k)$ \vspace{0.05cm}
  \STATE $\mathbf{Y}^{k+1}=\mathbf{Y}^k + \mu_k(\widehat{\mathbf{H}}_u - \mathbf{L}_u^{k+1} - \mathbf{S}_u^{k+1})$ \vspace{0.05cm}
  \STATE $\mu_k\gets\min \{ \rho \mu_k, \mu_{\max} \}$ \vspace{0.05cm} 
  \STATE Output: $\mathbf{L}_u$ and $\mathbf{S}_u$
\end{algorithmic}
\end{algorithm}

\subsection{Convergence Analysis of the A-RPCA Algorithm} \label{analysis}

The convergence of the proposed A-RPCA algorithm \eqref{algo1} is considered here. Let us write the optimization problem as 
\begin{equation}
\min_{\mathbf{L}_u, \, \mathbf{S}_u}\;
f(\mathbf{L}_u) + g(\mathbf{S}_u)
\quad
\text{s.t.}\quad
\mathbf{L}_u + \mathbf{S}_u = \widehat{\mathbf{H}}_u,
\label{eqARPCA}
\end{equation}
where
$f(\mathbf{L}_u) = \|\mathbf{L}_u\|_* + \gamma\|\mathbf{L}_u - \beta\mathbf{L}_1\|_F^2$ and $g(\mathbf{S}_u) = \lambda\|\mathbf{S}_u\|_1$.

The sequence $\{(\mathbf{L}_u^k,\mathbf{S}_u^k,\mathbf{Y}^k)\}$ in the form of ADMM updates is given as follows:
\begin{equation}
\mathbf{L}_u^{k+1} = \arg\min_{\mathbf{L}_u}\;
f(\mathbf{L}_u)
+ \tfrac{\mu_k}{2}\big\|\mathbf{L}_u - (\widehat{\mathbf{H}}_u - \mathbf{S}_u^k + \tfrac{1}{\mu_k}\mathbf{Y}^k)\big\|_F^2,
\end{equation}
\begin{equation}
\mathbf{S}_u^{k+1} = \arg\min_{\mathbf{S}_u}\;
g(\mathbf{S}_u)
+ \tfrac{\mu_k}{2}\big\|\mathbf{S}_u - (\widehat{\mathbf{H}}_u - \mathbf{L}_u^{k+1} + \tfrac{1}{\mu_k}\mathbf{Y}^k)\big\|_F^2,
\end{equation}
\begin{equation}
\mathbf{Y}^{k+1} = \mathbf{Y}^k + \mu_k(\widehat{\mathbf{H}}_u - \mathbf{L}_u^{k+1} - \mathbf{S}_u^{k+1}).
\label{yk_1}
\end{equation}
By first-order optimality of these proximal subproblems, we obtain:
\begin{equation}
\left\{
\begin{aligned}
0 &\in \partial f(\mathbf{L}_u^{k+1}) - \mathbf{Y}^k - \mu_k(\widehat{\mathbf{H}}_u - \mathbf{L}_u^{k+1} - \mathbf{S}_u^k) \\
0 &\in \partial g(\mathbf{S}_u^{k+1}) - \mathbf{Y}^k - \mu_k(\widehat{\mathbf{H}}_u - \mathbf{L}_u^{k+1} - \mathbf{S}_u^{k+1}). 
\end{aligned}
\right.
\end{equation}
Substituting the update of equation (\ref{yk_1})
yields:
\begin{equation}
\left\{
\begin{aligned}
& \mathbf{Y}^{k+1} \in \partial f(\mathbf{L}_u^{k+1}) + \mu_k(\mathbf{S}_u^{k+1}-\mathbf{S}_u^k),\\
& \mathbf{Y}^{k+1} \in \partial g(\mathbf{S}_u^{k+1}).
\end{aligned}
\right.
\label{equY}
\end{equation}

\begin{lemma} \label{lemY}
The sequence $\{\mathbf{Y}^k\}$ is bounded.
\end{lemma}
\begin{proof}
From equation (\ref{yk_1}) and (\ref{equY}), we have
\begin{equation}
\mathbf{Y}^{k+1}
= \mathbf{Y}^k + \mu_k(\widehat{\mathbf{H}}_u - \mathbf{L}_u^{k+1} - \mathbf{S}_u^{k+1})
\in \lambda\,\partial\|\mathbf{S}_u^{k+1}\|_1,
\end{equation}
$\implies$ every entry of $\mathbf{Y}^{k+1}$ is $\le \lambda$ and $\|\mathbf{Y}^{k+1}\|_\infty\le \lambda$.
Then $\|\mathbf{Y}^{k+1}\|_F^2=\sum_{ij}(Y^{k+1}_{ij})^2\le N_sN_b\,\|\mathbf{Y}^{k+1}\|_\infty^2$, 
and $\|\mathbf{Y}^{k+1}\|_F^2 \le \lambda^2 N_sN_b$.
\end{proof}

\begin{lemma} \label{lemLS}
The sequence $\{(\mathbf{L}_u^k,\mathbf{S}_u^k)\}$ is bounded.
\end{lemma}
\begin{proof}
See Appendix \ref{appC}
\end{proof}

\begin{theorem}
Any accumulation point $(\mathbf{L}_u^*,\mathbf{S}_u^*,\mathbf{Y}^*)$ of
$\{(\mathbf{L}_u^k,\mathbf{S}_u^k,\mathbf{Y}^k)\}$ satisfies the KKT (Karush-Kuhn-Tucker) conditions:
\begin{equation}
\left\{
\begin{aligned}
& \mathbf{L}_u^*+\mathbf{S}_u^*=\widehat{\mathbf{H}}_u,\\
& \mathbf{Y}^*\in \partial f(\mathbf{L}_u^*),\\
&\mathbf{Y}^*\in \partial g(\mathbf{S}_u^*),
\end{aligned}
\right.
\end{equation}
and then is a stationary point of (\ref{eqARPCA}).
\end{theorem}
\begin{proof}
    From lemmas \ref{lemY} and \ref{lemLS}, the sequence $\{(\mathbf{L}_u^k, \mathbf{S}_u^k, \mathbf{Y}^k)\}$ is bounded, this ensures existence a accumulation point $(\mathbf{L}_u^*,\mathbf{S}_u^*,\mathbf{Y}^*)$. 
 Also $
    \sum_{k=0}^{\infty}\mu_k\|\mathbf{r}^{k+1}\|_F^2 < \infty \implies 
\mathbf{P}^k \to 0 $, then $\mathbf{L}_u^*+\mathbf{S}_u^*=\widehat{\mathbf{H}}_u$.

From equation (\ref{equY}),
\(
\mathbf{Y}^{k+1}\in \partial g(\mathbf{S}_u^{k+1}) \) and 
\(\mathbf{Y}^{k+1}\in \partial f(\mathbf{L}_u^{k+1}) + \mu_k(\mathbf{S}_u^{k+1}-\mathbf{S}_u^k)
\). Since $\mu_k(\mathbf{S}_u^{k+1}-\mathbf{S}_u^k)\to 0$, passing to the limit yields $\mathbf{Y}^*\in \partial g(\mathbf{S}_u^*)$ and $\mathbf{Y}^*\in \partial f(\mathbf{L}_u^*)$. The KKT conditions are satisfied. Therefore $(\mathbf{L}_u^*,\mathbf{S}_u^*,\mathbf{Y}^*)$ is a stationary point of (\ref{eqARPCA}).
\end{proof}


\section{CSI-PLA Using Information Reconciliation} 
\label{auth_setup}

After the CSI preprocessing, the preprocessed data are fed into the input of the reconciliation-based PLA scheme. Thus, $\mathbf{L}_1(t)$ and $\mathbf{L}_u(t+1)$ are reshaped and then quantized to obtain $\mathbf{q}_1(t)$ and $\mathbf{q_u}(t+1)$ $\in \{0,1\}^{1 \times N}$, where  $\mathbf{q}_1(t)$ and $\mathbf{q_u}(t+1)$ correspond to one sample, and $N$ is the codelength. We use the Lloyd-Max quantizer \cite{lloyd1982, max1960}, which is effective to design optimal quantizers. We also consider Gray coding to map the quantizer outputs into bits.

\subsection{Information Reconciliation}

The vectors $\mathbf{q}_1(t)$ and $\mathbf{q_u}(t+1)$ are mapped to the decoder outputs $\mathbf{r}_1(t)$ and $\mathbf{r}_u(t+1)$. The information reconciliation scheme uses a decoding process based on the principle of Slepian-Wolf decoding as in \cite{ArikanPolar}, where $\mathbf{q}_a(t+1)$ is decoded using the side information (helper data or syndrome) $S$ derived from the code design and the quantized vector $\mathbf{q}_1(t)$ from the enrollment phase at time $t$. In this work, we adopted polar codes \cite{ArikanPolar} as the error-correcting codes. Therefore, the side information contains the message and the frozen bit positions, and the $\mathrm{CRC}$ bits that are used for the decoding of $\mathbf{q}_u(t)$. 

This approach ensures that the correlation between the CSI across the training and the authentication phases is appropriately exploited in order to distinguish different users.
The reconciliation is expected to yield (i) identical or almost equal outputs in the normal scenario and (ii) distinctly different outputs in the alternative scenario. Specifically, since the correlation between the CSI at time $t$ and $t+1$ is much lower in the alternative situation, the decoder using the side information, fails to produce accurate outcomes. In contrast, the decoder produces accurate outputs in normal situations. The polar code construction and decoding are described in the following subsection.

\subsection{Polar Codes} 

In this subsection, we detail the construction and decoding of the polar codes adopted in the proposed PLA scheme.

\subsubsection{Code construction}

We decided to construct the polar codes using Gaussian approximation (GA) \cite{trifonov2012,pga}, which assumes that log-likelihood ratios (LLRs) remain Gaussian distributed through the polar transformation.  This is a more accurate construction rather than an upper bound as in \cite{arikan2009} where the construction uses the binary erasure channel (BEC) as a surrogate for the design channel, ranking polarized bit-channels via the Bhattacharyya recursion. The idea of GA is to evolve the densities and estimate the precise reliability of each channel. For AWGN with variance $\sigma^2$, the channel LLR is modeled as
\(
\text{LLR} \sim \mathcal{N}(\mu,2\mu)\) and \(\mu_0 = \frac{2}{\sigma^2},  
\)
where $\mu_0$ is the initial LLR mean. 
GA propagates only the mean LLR $\mu$ through the polarization process, yielding $\mu_i$ for each synthesized bit-channel $W^{(i)}$.
Let $\phi(\cdot)$ denote the GA check-node function~\cite{trifonov2012}, with an accurate empirical fit
\begin{equation}
\phi(\mu)\approx 
\begin{cases}
\exp(-0.4527\,\mu^{0.86}+0.0218), & 0\leq \mu \leq 10, \\
\sqrt{\tfrac{\pi}{\mu}}\left(1-\tfrac{10}{\mu}\right)\exp(-\mu/4), & \mu > 10.
\end{cases}
\end{equation}
The polarized outputs are given by 
\begin{equation}
\mu^{-} = \phi^{-1}\!\big(1-(1-\phi(\mu))^2\big),\qquad \mu^{+}=2\mu,
\end{equation}
where $\mu^{-}$ is the GA-estimated reliability of the bad channel and $\mu^{+}$ is the GA-estimated reliability of the good channel. Then the reliabilities of the polar code are found by sorting the $\mu_i$ values. The larger $\mu_i$ correspond to the more reliable bit-channels. In this work, we sort them from the least reliable to most reliable.


\subsubsection{Decoding}

In this part, $\mathbf{q}_u(t+1)$ is decoded using the side information $S$ from time $t$, which contains the frozen bits and the $\mathrm{CRC}$ bits. To enhance the performance of the polar codes in finite blocklengths, a $\mathrm{CRC}$-aided successive cancellation list decoding (SCL) \cite{Shakiba2},\cite{Tal} is applied, where the list size is denoted by $\ell_s$. 
During the decoding process, the $\mathrm{CRC}$-aided successive cancellation list decoder maintains $\ell_s$ decoding paths simultaneously. At each decoding stage, the decoder extends all active paths by considering both possible bit decisions (0 and 1) and the path metrics are updated accordingly to reflect their likelihoods. To manage complexity, only the $\ell_s$ most probable paths are retained, while the others are discarded. 
The inclusion of the $\mathrm{CRC}$ enhances the decoder's ability to distinguish between valid and invalid codewords, thereby significantly improving the reliability of the decoding process.

Now let us give a brief description of the decoding process using SCL with CRC by assuming that the polar code has length $N$. Considering the dimension (message length) $K$, the number of frozen bits (bits under pure-noise) is $N-K$.
Let us denote the input bits and the received bits as \( u_1^N = (u_1, \ldots, u_N) \) and \( y_1^N = (y_1, \ldots, y_N) \), respectively. The received bits here correspond to our quantized vectors of the CSI and the inputs that will be estimated with the decoding correspond to the reconciled vectors. In successive cancellation (SL) decoding, the information bits are decoded sequentially where each bit is estimated using maximum-likelihood (ML) decoding. The information bits $i$ are then decoded by calculating the likelihood ratio as follows \cite{arikan2009}:
\begin{equation}
L_N^{(i)}(y_1^n, \hat{u}_1^{i-1}) = \frac{W_N^{(i)}(y_1^N, \hat{u}_1^{i-1}|0)}{W_N^{(i)}(y_1^N, \hat{u}_1^{i-1}|1)},
\end{equation}
where $W_N^{(i)}(\cdot|\cdot)$ is the probability density. The decision is given by
\begin{equation}
\hat{u}_i =
\begin{cases}
0, & \text{if } L_N^{(i)}(y_1^N, \hat{u}_1^{i-1}) \ge 1 \\
1, & \text{otherwise.}
\end{cases}
\label{deci}
\end{equation}
In SCL decoding, instead of deciding to set the values to $0$ or $1$, the decoding path is split into two paths and the new paths are also split into two paths and so on until the maximum number of paths allowed is reached, which corresponds to the list size $\ell_s$. During this process, the most likely paths are kept. 

In the CRC-aided SCL decoding method, a part of the $K$ information bit are used to calculate the $m$ CRC bits. Thus, the first $K-m$ unfrozen bits become the information bits and the $m$ remaining bits are CRC bits. This helps at the final stage of the decoding to choose the transmitted bit by first removing the paths that have incorrect CRC bits. Among the surviving paths, the codeword that satisfies the $\mathrm{CRC}$ constraint and has the highest likelihood is selected as the final decoded output.

\subsubsection{Performance metrics} 
After the reconciliation, hypothesis testing is applied to differentiate the users. Two hypothesis $H_0$ and $H_1$ corresponding respectively to the normal case where the targeted user is transmitting and the alternative case when the user is not transmitting, are defined as follows:
\begin{equation}
    \centering
    \left\{ \begin{aligned}
H_0: \hspace{0.2cm}&\eta=\mathcal{H}_d\left(\mathbf{r}_1(t), \mathbf{r}_1(t+1)\right) \leq \eta_{t h} \\
H_1: \hspace{0.2cm} &\eta=\mathcal{H}_d\left(\mathbf{r}_1(t), \mathbf{r}_2(t+1)\right)>\eta_{t h}
    \end{aligned}
\right.
\label{eq}
\end{equation}
A bitwise comparison between $\mathbf{r}_1(t)$ and $\mathbf{r}_u(t+1)$, $u \in \{1, \, 2\}$, is a suitable choice for $\eta$. Thus, $\mathcal{H}_d(\cdot)$ is the Hamming distance between $\mathbf{r}_1(t)$ and $\mathbf{r}_u(t+1)$.
\begin{equation}
\eta = \mathcal{H}_d\left(\mathbf{r}_1(t), \mathbf{r}_u(t+1)\right) = \sum_{j=1}^K\left|r_{1,j}(t)-r_{u,j}(t+1)\right|
\label{eta}
\end{equation}
We can then evaluate the performance with the bit error probability, the receiver operation characteristic (ROC) curves and the detection and miss-detection probabilities.

\section{Numerical Results} 
\label{results}
In this section, we assess the proposed PLA scheme and A-RPCA algorithm via simulations using i) synthetic data in subsection \ref{res_A} and ii)  a real massive multiple input multiple output (mMIMO) outdoor dataset at FR1 (2.18 GHz) collected during a measurement campaign from Nokia in subsection \ref{res_B}. In each  subsection, we present numerical results for both CSI prepropcessing and CSI-PLA using reconciliation. We compare our proposed work with the CSI without preprocessing denoted by NoPre and previous schemes based on PCA, RPCA, AE and ReProCS. For the AE, table \ref{table_AE} gives the detailed structure of neural network.

\begin{table}[!t]
\centering
\caption{Network structure of the AE}
\renewcommand{\arraystretch}{1.2}
\begin{tabular}{|c|c|c|}
\hline
Layer & Dimensions & Activation \\
\hline
Input & $N_b$  & linear \\
\hline
1 & 200 & tanh \\
\hline
2 &  100 & softplus \\
\hline
3 & 50 & tanh  \\
\hline
4 (intermediate) & $D_{AE} = 32$ &  linear \\
\hline
5 & 50 &  relu\\
\hline
6 & 100 & softplus \\
\hline
7 & 200 & tanh \\
\hline
Output & $N_b$ & sigmoid \\
\hline
\end{tabular}
\label{table_AE}
\end{table}

\subsection{Synthetic Data} \label{res_A}
{
In this part, we present numerical results for synthetic data for the case where the channel follows a Rician distribution as mentioned in section \ref{sys_model}. 
Unless otherwise specified, the simulation parameters are defined as follows: $N_b = 32$, $\beta = 0.9$, the Code rate = $0.1$ and the signal-to-noise ratio (SNR) is set to $10\;dB$.

\subsubsection{CSI Preprocessing} \label{sb_prepro_synth}
The performance of the CSI preprocessing in evaluated by computing the Pearson correlation coefficient and the bit mismatch rate (BMR).

Fig. \ref{heat_synth} present the heatmap of the correlation between the estimated CSI at time $t$ and $t+1$ under both hypothesis $H_0$ (yellow square) and $H_1$ (red square) for A-RPCA in Fig. \ref{heat_arpca} and ReProCS in Fig. \ref{heat_reprocs}. This result show how the proposed A-RPCA can impact the correlation, as we have the value almost $1$ compare to RePoCS which is around $0.7$.

In Fig. \ref{bmr}, we show the bit mismatch rate (after quantization) as a function of the Rician $K$-factor $K = K_1 = K_2$. This allows assessment of the impact of the preprocessing when the alternative user Alice $2$ experiences the same LoS as Alice $1$. We can see that the proposed A-RPCA preprocessing performs much better than A-RPCA and the CSI without preprocessing, where $K = 0$ corresponds to a NLoS scenario. When $K$ increases, even the BMR between Alice $1$ (time $t$) and Alice $2$ (time $t+1$) is very low and almost equal to that of Alice $1$ (time $t$) and Alice $2$ (time $t+1$). Consequently, the improvement is obvious $K_1 = K_2 = 0$.

{
The BMR is also shown in Table \ref{tab_bmr_sy} for different SNR values for the case in which $K_1 = K_2 = 0$. Under hypothesis $H_1$, all the values are around $0.5$ for all the preprocessing methods. In contrast, under $H_0$, A-RPCA significantly outperforms the other schemes. For example, for SNR $ = \SI{10}{dB}$, the BMR drops to $0.08$ for A-RPCA while the other schemes obtain a performance close to NoPre where BMR $= 0.19$.
}

\begin{figure}[!t]
  \centering
  \subfloat[\footnotesize A-RPCA \label{heat_arpca}]{
    \includegraphics[width=0.6\linewidth]{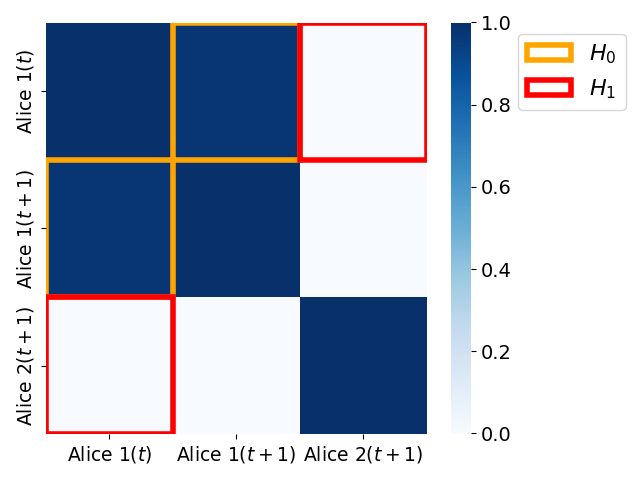} 
  } \hfill 
  \subfloat[\footnotesize ReProCS \label{heat_reprocs}]{
    \includegraphics[width=0.6\linewidth]{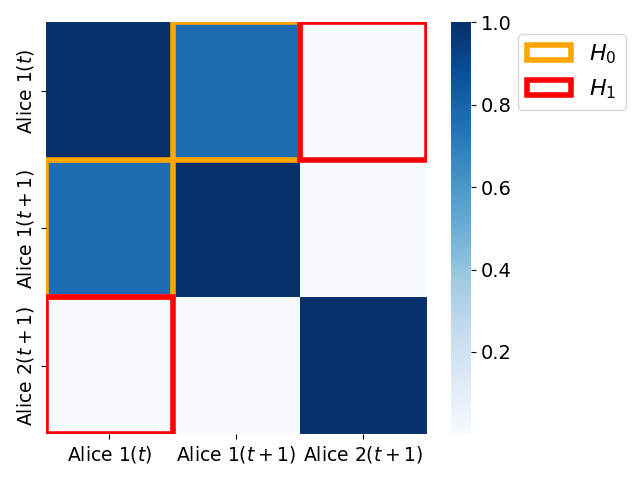}
    }
  \caption{Heatmap representing the correlation coefficient between the CSI after preprocessing at $t$ and $t+1$: the hypothesis $H_0$ is shown by the yellow and $H_1$ by the red square. The SNR $ = \SI{10}{\decibel}$.}
  \label{heat_synth}
\end{figure}
}

\begin{figure}[!t]
\centering
\includegraphics[width=0.4\textwidth, height=0.25\textwidth]{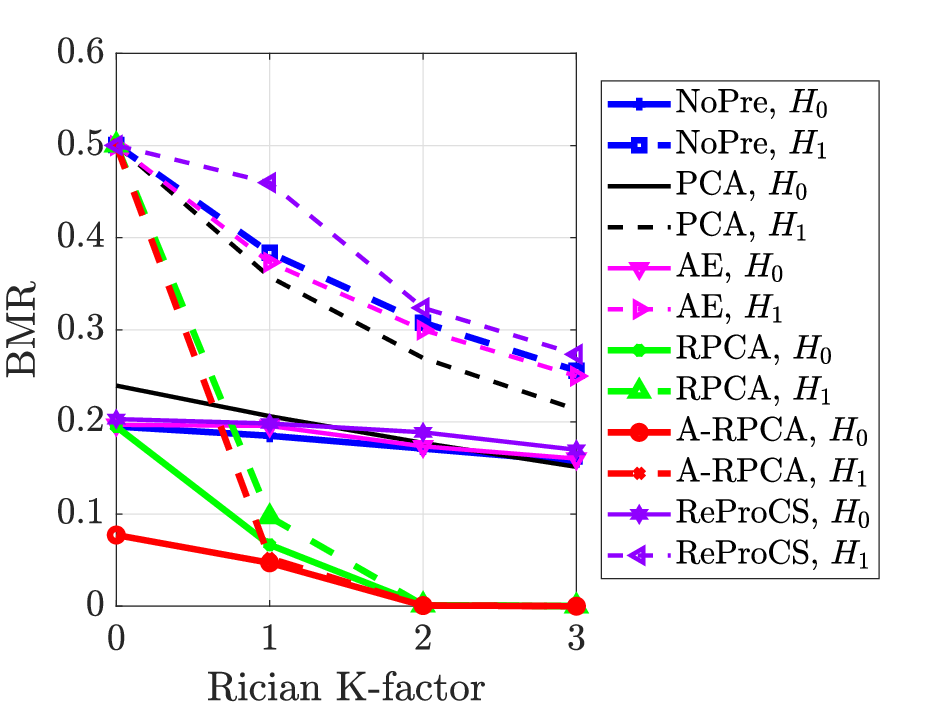} 
\caption{Bit mismatch rate: $\mathrm{SNR} = \SI{10}{\decibel}$, $N = 128$, $n = 1$ bit} 
\label{bmr}
\end{figure}

\begin{table*}[!t]
\centering
\caption{Bit mismatch rate: $K_1 = K_2 = 0$, $D_{AE} = 32$ (for AE), $D_{PCA} = 10$}
\renewcommand{\arraystretch}{1.2}
\begin{tabular}{|c|c|c|c|c|c|c|c|c|c|c|c|c|}
\hhline{~|-|-|-|-|-|-|-|-|-|-|-|-|}
\multicolumn{1}{c|}{} & \multicolumn{2}{c|}{NoPre} & \multicolumn{2}{c|}{PCA} & \multicolumn{2}{c|}{AE} & \multicolumn{2}{c|}{ReProCS} & \multicolumn{2}{c|}{RPCA} & \multicolumn{2}{c|}{A-RPCA} \\
\hline
SNR [$\SI{}{\decibel}$] & $H_0$ & $H_1$ & $H_0$ & $H_1$ & $H_0$ & $H_1$ & $H_0$ &$H_1$ & $H_0$ & $H_1$ & $H_0$ & $H_1$ \\
\hline
$5$ & $0.26$ & $0.50$ & $0.30$ & $0.50$ & $0.27$ & $0.50$ & $0.27$ & $0.50$& $0.26$ & $0.50$ &$0.09$ & $0.50$ \\
\hline
$10$ &  $0.19$ & $0.50$ & $0.24$ & $0.50$ & $0.20$ & $0.50$  & $0.20$& $0.50$ &$0.19$ & $0.50$ & $0.08$&$ 0.50$ \\
\hline
$15$ & $0.16$ & $0.50$ &  $0.21$ & $0.50$ & $0.17$ & $0.50 $ & $0.17$ &$0.50$ & $0.16$ & $0.50$ & $0.07$ &  $0.50$ \\
\hline
\end{tabular}
\label{tab_bmr_sy}
\end{table*}

\subsubsection{PLA Using Information Reconciliation}
In this part we assess the performance of PLA using information after preprocessing. Performance metrics like detection probability, false alarm probability and error probabilities after reconciliation are assessed. 

In Fig. \ref{roc_synth}, a receiver operating characteristic (ROC) curve is shown. The proposed A-RPCA algorithm has very high detection probability (almost $100 \%$) and this confirms the results obtained in Fig. \ref{bmr}. 

\begin{figure} 
\centering
\includegraphics[width=0.35\textwidth, height=0.25\textwidth]{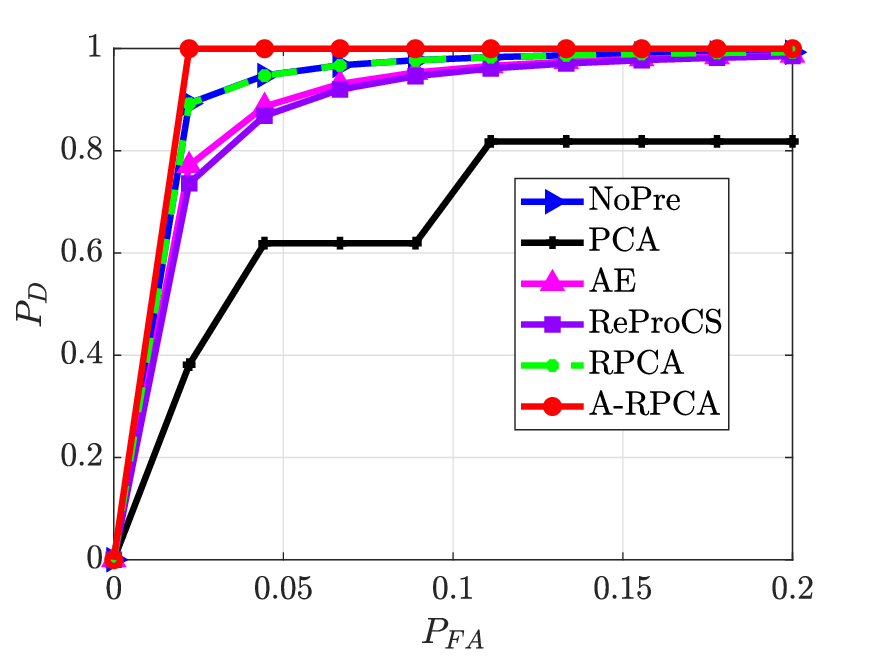}
\caption{$P_D$ vs $P_{FA}$: $K_1 = K_2 = 0$, $\mathrm{SNR} = \SI{10}{\decibel}$, Code rate = $0.1$, $N = 128$.} 
\label{roc_synth}
\end{figure}

An analysis of the code rate is carried out in Fig. \ref{error_prob} and this shows that the error probability is lower under the hypothesis $H_0$, as expected. In particular, A-RPCA has very good performance as compared to the other schemes when the code rate increases.
\begin{figure}[!t]
\centering
\includegraphics[width=0.35\textwidth, height=0.25\textwidth]{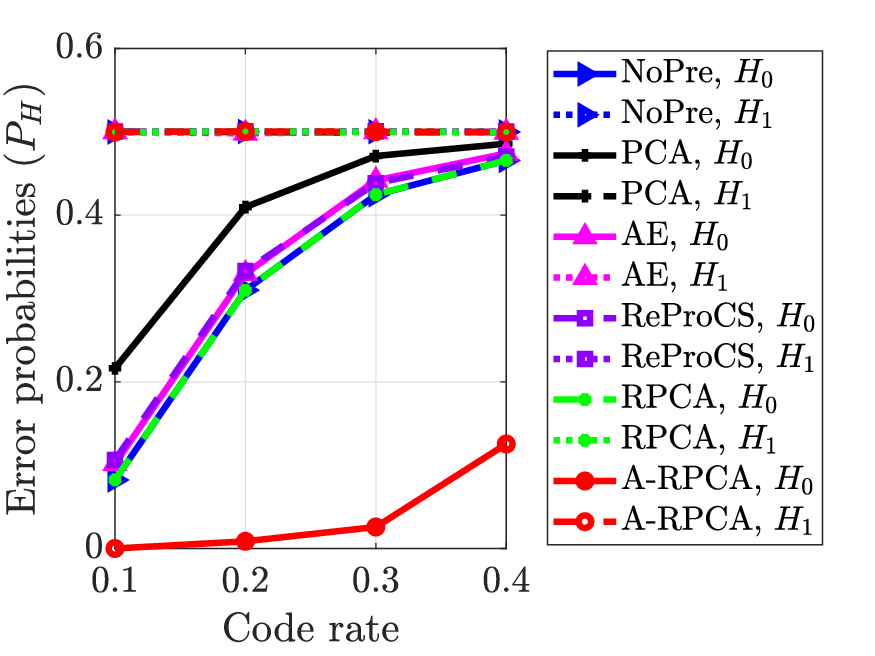} 
\caption{Error probability after reconciliation under $H_0$ and $H_1$: $\mathrm{SNR} = \SI{10}{\decibel}$.} 
\label{error_prob}
\end{figure}

\subsection{Results on Real Dataset} \label{res_B}
In this section, we present results on a real mMIMO CSI measurement campaign carried out on the Nokia campus in Stuttgart \cite{Shehzad}. Our purpose here is to distinguish a target track from the others, as shown in Fig. \ref{nokia_map}.
\begin{figure*}
\centering
\includegraphics[width=0.66\textwidth]
{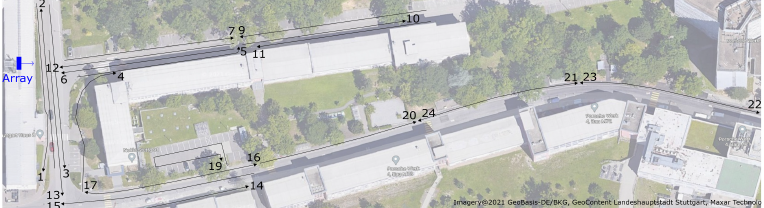}  
\caption{Map of the Nokia campus in Stuttgart, Germany. The transmitter array location is represented by the blue bar. The user equipment (UE) is moving along different tracks represented by the black lines and the direction is presented by the arrow.}   
\label{nokia_map}
\end{figure*}

The measurement environment consists of buildings approximately 15~m high, primarily arranged along the streets. These buildings acted as both reflectors and obstructions for the radio waves transmitted from the antenna array, which was mounted on the rooftop of one of the buildings. The site thus included both LoS and NLoS propagation.

The transmit array comprised 64 elements arranged in four rows of 16 single-polarized patch antennas, with a horizontal spacing of $\lambda/2$ and a vertical separation of $\lambda$. The array transmitted 64 orthogonal pilot signals across both time and frequency domains at a carrier frequency of 2.18~GHz using OFDM waveforms. The pilot structure was designed such that 50 separate subbands, each consisting of 12 consecutive subcarriers, were sounded within 0.5~ms, during which the channel was assumed to remain invariant in time. Pilot bursts were transmitted continuously with a periodicity of 0.5~ms.

The receiver setup emulating a user equipment (UE) consisted of a single monopole antenna placed at a height of 1.5~m, connected to a Rohde \& Schwarz TSMW receiver and a Rohde \& Schwarz IQR hard disk recorder, which continuously captured the received baseband signal. 
Both the transmitter and receiver were frequency-synchronized via GPS. Then, during measurements, a receiver cart was moved along multiple predefined routes, as shown in Fig. \ref{nokia_map}, at a walking speed of approximately 5~km/h, corresponding to a spatial sampling 
distance below 0.1~mm. In post-processing, a 64-dimensional channel vector was extracted for each pilot burst and subband. 

Fig. \ref{trackss} shows the track segmentation used in this work for the pairs of tracks we aim to differentiate (referred to Alice $1$ and Alice $2$). One segment represents the track of interest (Alice $1$) to be detected, where the green squares represent the user position at time $t$ and the orange squares correspond to the user locations at time $t+1$. The red squares represent the user location at $t+1$ on the alternative track (Alice $2$) from which we aim to distinguish the user of interest. Different communication scenarios are considered, where (Alice $1$)  and (Alice $2$) are on LoS or NLoS tracks. 
\begin{figure}
\centering
\includegraphics[width=0.5\textwidth]{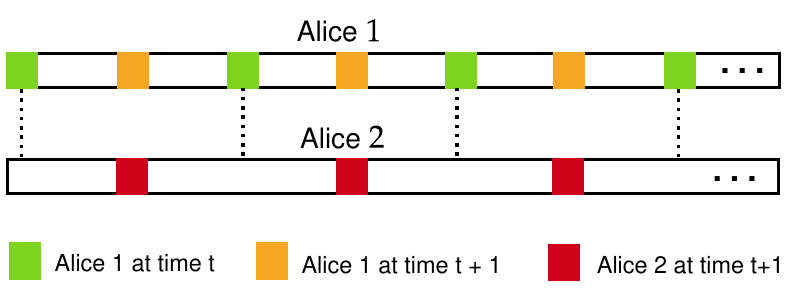}  
\caption{Tracks of Alice 1 and Alice 2 to be differentiated. Each Alice is on a single track because the data corresponds to a walking user. The green and orange colors represent respectively Alice $1$ locations at time $t$ and $t+1$. The red color is Alice $2$ location at time $t+1$.}   
\label{trackss}
\end{figure}

In the following, unless otherwise specified, the parameters are defined as follows: code rate $ = 0.1$, codelength $N = 128$, signal-to-noise ratio $\mathrm{SNR} = \SI{10}{\decibel}$, number of subtracks $N_s = 46$, $ n = 1$ bit quantizer, for tracks (t$6$, t$1$), that is, t$6$ for Alice $1$ and t$1$ for Alice $2$.

\subsubsection{CSI Preprocessing}  
As in section \ref{res_A}, the performance is assessed by showing the correlation heatmap and the BMR, respectively, in Fig. \ref{heat_nokia} and Fig. \ref{bmr_nokia} for SNR $= \SI{5}{\decibel}$. The correlation coefficient is almost equal to $1$ for A-RPCA while it is around $0.5$ for ReProCS. Concerning the BMR, the improvement with A-RPCA is also shown, in particular, under $H_0$. Under $H_1$, all the preprocessing schemes give a BMR around $0.5$, as well as NoPre.

\begin{figure}[!t]
  \centering
  \subfloat[\footnotesize A-RPCA \label{heat_arpca_nokia}]{
    \includegraphics[width=0.6\linewidth]{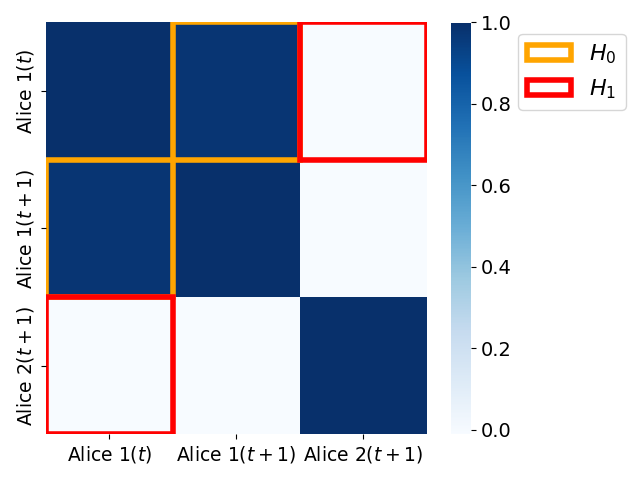} 
  } \hfill 
  \subfloat[\footnotesize ReProCS \label{heat_reprocs_nokia}]{
    \includegraphics[width=0.6\linewidth]{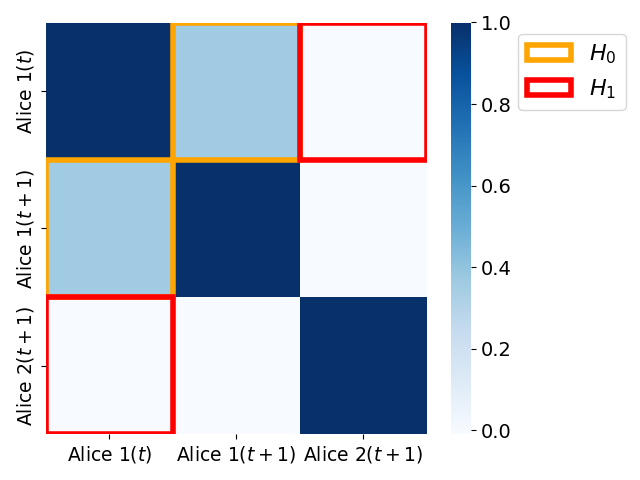}
    }
  \caption{Heatmap representing the correlation coefficient between the CSI after preprocessing at $t$ and $t+1$: the hypothesis $H_0$ is shown by the yellow and $H_1$ by the red square. The SNR $ = \SI{5}{\decibel}$.}
  \label{heat_nokia}
\end{figure}

\begin{figure}[!t]
\centering
\includegraphics[width=0.35\textwidth, height=0.25\textwidth]{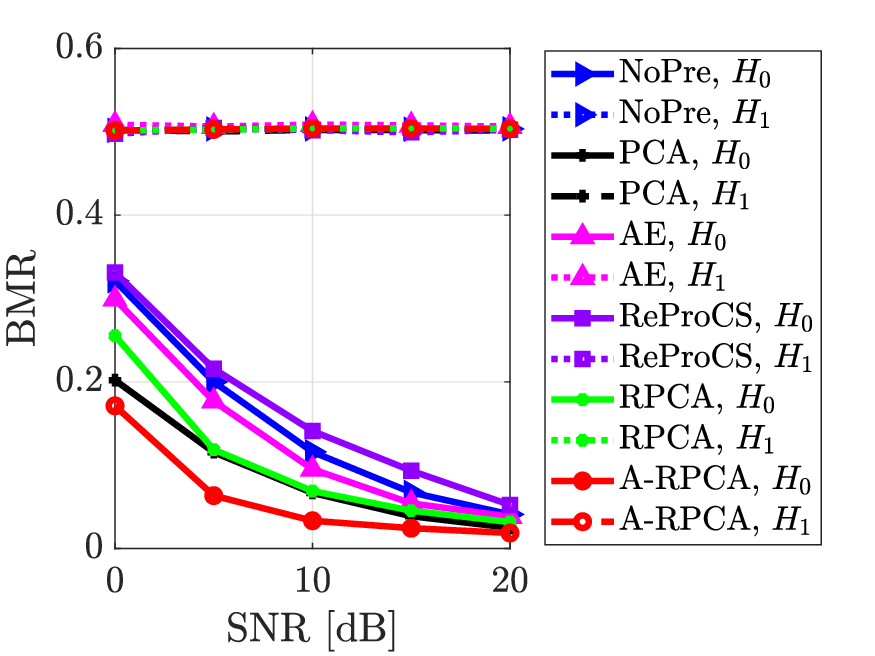} 
\caption{Bit mismatch ratio under $H_0$ and $H_1$: $\mathrm{SNR} = \SI{10}{\decibel}$.} 
\label{bmr_nokia}
\end{figure}


\begin{table}[!t]
\centering
\caption{$P_D$ vs Code rate: $\mathrm{SNR} = \SI{10}{\decibel}$, $P_{FA} = 0.05$}
\renewcommand{\arraystretch}{1.2} 
\begin{tabular}{|c|c|c|c|c|c|c|}
\hhline{|-|-|-|-|-|-|-|}
Rate & NoPre & PCA & AE & ReProCS & RPCA & A-RPCA\\
\hline
$0.1$ & $1$ & $1$ & $1$ & $1$ & $1$ & $1$  \\
\hline
$0.2$ & $0.94$ & $0.97$ & $0.99$ & $0.90$ & $1$ & $1$  \\
\hline
$0.3$ & $0.92$ & $0.97$ & $0.95$ & $0.86$ & $0.99$ & $1$  \\
\hline
$0.4$ & $0.88$ & $0.96$ & $0.87$ & $0.72$ & $0.98$ & $0.99$  \\
\hline
\end{tabular}
\label{tab_bmr_nokia}
\end{table}

\subsubsection{PLA Using Information Reconciliation}  
First, the error probability after the reconciliation is shown for all the schemes under $H_0$ and $H_1$ as a function of the code rate. As expected, the proposed scheme outperforms prior schemes under $H_0$, particularly after $0.3$. PCA and RPCA give also very good results for codes rates less than $0.3$ 
\begin{figure}[!t]
\centering
\includegraphics[width=0.35\textwidth, height=0.25\textwidth]{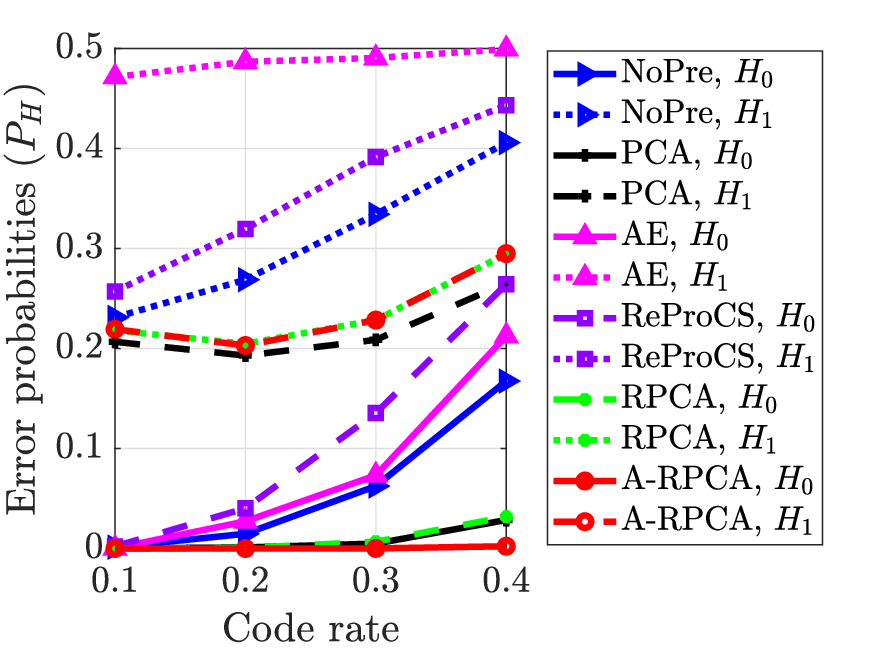} 
\caption{Error probability after reconciliation under $H_0$ and $H_1$: $\mathrm{SNR} = \SI{10}{\decibel}$.} 
\label{error_prob_nokia}
\end{figure}
 
The ROC curve is presented in Fig. \ref{roc61_02} for a code rate $= 0.2$ and SNR$ = \SI{10}{\decibel}$. This demonstrates the effectiveness of the reconciliation scheme in real scenarios. All cases, including the CSI without preprocessing, perform quite well, while A-RPCA has detection probabilities equal to $1$ even for false alarm ratios close to $0$. The detection probability is also evaluated for different values for the SNR in table \ref{tab_pd_snr_nokia}. RPCA and the proposed A-RPCA algorithm perfom very well for all SNR values but A-RPCA is more robust for SNR $= \SI{5}{\decibel}$. 

The ROC curve is also presented in Fig. \ref{roc61_02_2} where we compare the performance of the authentication scheme with and without reconciliation denoted by "no recon". For all preprocessing schemes presented in this figure, the reconciliation is performing very well compare to the authentication without reconciliation. This result confirms also the effectiveness of the CSI-PLA using information reconciliation using real dataset.

\begin{figure} [!t]
\centering
\includegraphics[width=0.35\textwidth, height=0.25\textwidth]{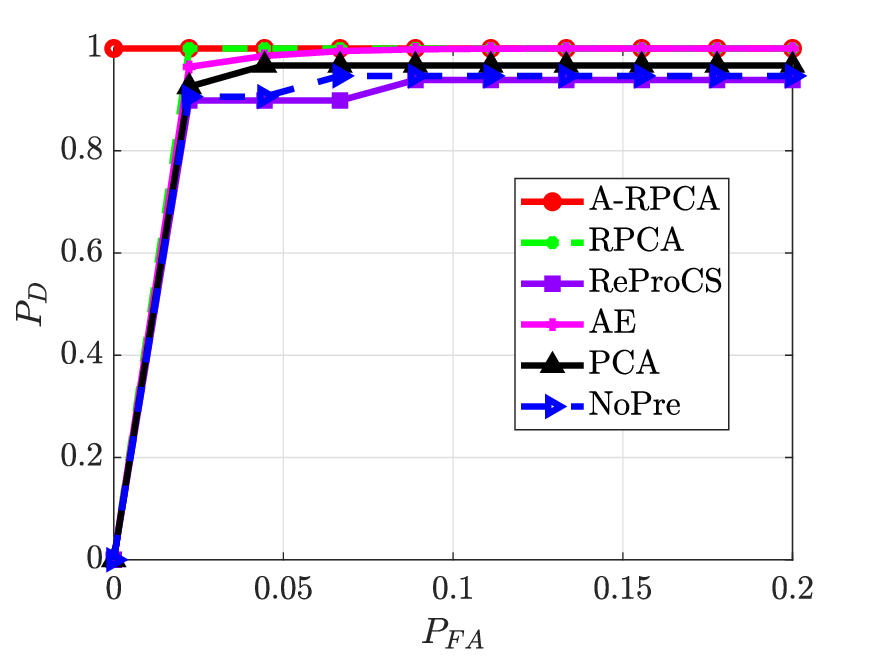}
 \caption{$P_D$ vs $P_{FA}$:  (LoS, NLoS) $= (6,1)$, $\mathrm{SNR} = \SI{10}{\decibel}$, $N = 128$, Code rate = $0.2$}   
 \label{roc61_02}
\end{figure}

\begin{figure} 
\centering
\includegraphics[width=0.35\textwidth, height=0.25\textwidth]{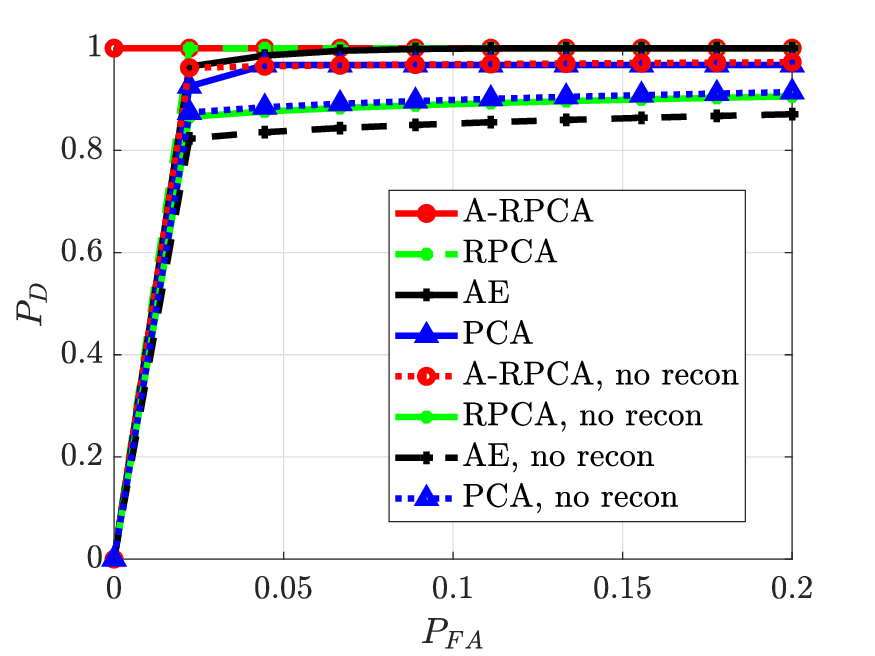}
 \caption{$P_D$ vs $P_{FA}$:  (LoS, NLoS) $= (6,1)$, $\mathrm{SNR} = \SI{10}{\decibel}$, $N = 128$, Code rate = $0.2$}   
 \label{roc61_02_2}
\end{figure}

\begin{table}[!t]
\centering
\caption{$P_D$ vs SNR: Code rate $ = 0.2$, $P_{FA} = 0.05$}
\renewcommand{\arraystretch}{1.2} 
\begin{tabular}{|c|c|c|c|c|c|c|}
\hhline{|-|-|-|-|-|-|-|}
SNR [dB] & NoPre & PCA & AE & ReProCS & RPCA & A-RPCA\\
\hline
$5$ & $0.36$ & $0.96$ & $0.36$ & $0.21$ & $0.97$ & $0.99$  \\
\hline
$10$ & $0.95$ & $0.97$ & $0.99$ & $0.90$ & $1$ & $1$  \\
\hline
$15$ & $0.99$ & $0.99$ & $1$ & $0.99$ & $1$ & $1$  \\
\hline
\end{tabular}
\label{tab_pd_snr_nokia}
\end{table}


The miss-detection and false alarm probabilities are evaluated as a function of the hypothesis testing threshold in  
Fig. \ref{proba6_1_3} where the code rate is equal 
to $0.3$. The figure confirms our previous results with ROC curves. The equal error rate (EER), which corresponds to the point where the false alarm and miss-detection probability cross, is lower (almost $0$) for A-RPCA as compared to NoPre and the other schemes. This confirms the higher accuracy obtained by A-RPCA in Fig. \ref{roc61_02}. 


\begin{figure}[!t]
\centering
\includegraphics[width=0.35\textwidth, height=0.25\textwidth]{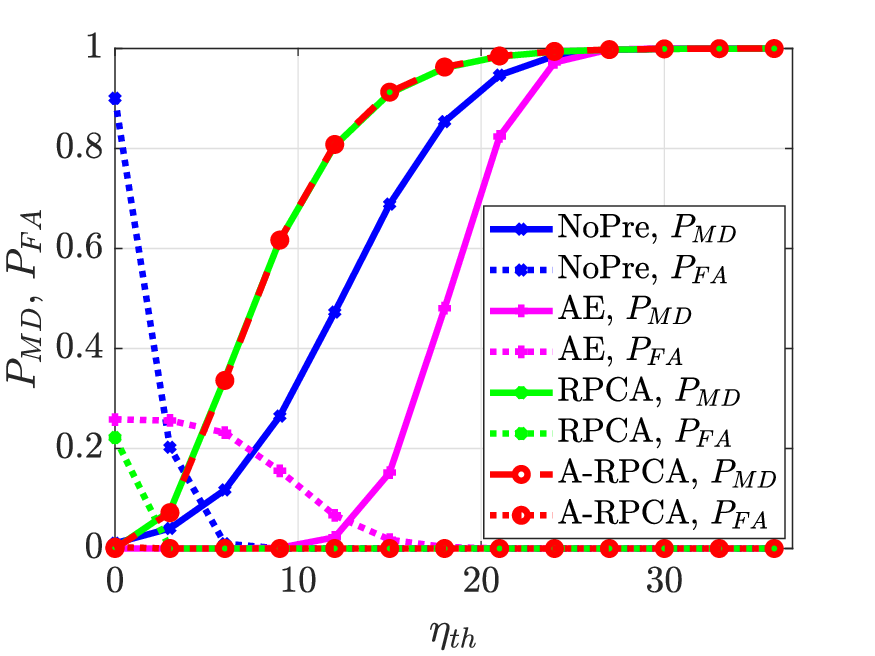}
\caption{$P_D$ and $P_{FA}$ vs $\eta_{th}$:  (LoS, NLoS) $= (6,1)$, $\mathrm{SNR} = \SI{10}{\decibel}$, Code rate = $0.3$, $N = 128$, $n = 1$ }   
\label{proba6_1_3}
\end{figure}



The performance is also assessed for different scenarios in Fig. \ref{r_tr}. Four cases are considered for the tracks of the pairs (Alice $1$, Alice $2$) as follows: (LoS, LoS) = (t$6$, t$11$), (LoS, NLoS) = (t$6$, t$1$), (NLoS, LoS) = (t$1$, t$6$) and (NLoS, NLoS) = (t$1$, t$13$). For a code rate equal to $0.2$, SNR $= \SI{5}{\decibel}$ and a false alarm probability $P_{FA} = 0.05$, A-RPCA performs very well with probability almost equal to $100 \%$ for all communication environments of Alice $1$ and Alice $2$.
\begin{figure}[!t]
\centering
\includegraphics[width=0.42\textwidth, height=0.3\textwidth]{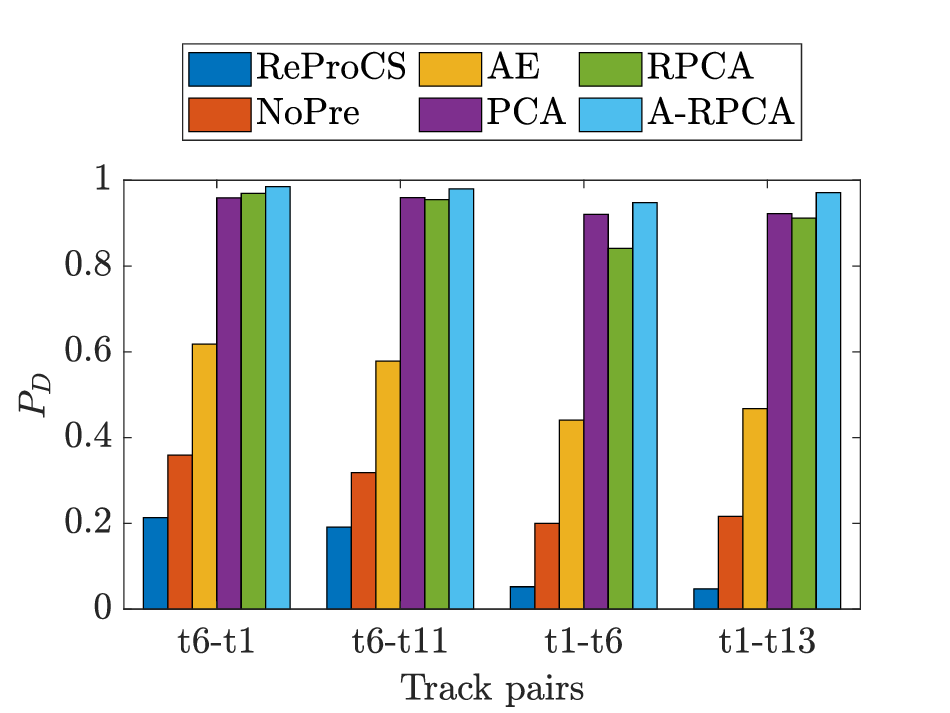}  \caption{Detection probability for different tracks: $P_{FA} = 0.
05$, $\mathrm{SNR} = \SI{5}{\decibel}$, Code rate = $0.2$, $N = 128$}   
\label{r_tr}
\end{figure}

\section{Conclusion} \label{conclude}
{
In this work, we have presented an adaptive CSI preprocessing scheme to enhance the accuracy of CSI-PLA and also evaluated the performance improvement when employing information reconciliation.  The proposed A-RPCA preprocessing technique was obtained by solving a TR-PCP optimization problem. Extensive sets of numerical results were presented for both synthetic data and real outdoor datasets. For both datasets, the preprocesing has been first assessed followed by CSI-PLA using information reconciliation. The  Numerical results showed that A-RPCA substantially improved the 
Pearson correlation coefficient and the detection probability for very small false alarm probabilities and outperformed accross all PLA metrics other state-of-the-art approaches, notably using PCA, RPCA, AE and ReProCS.
}

\section*{Acknowledgments}
A. K. A. Passah has been supported by SRV ENSEA, CNPq and the COST Action CA22168 6G-PHYSEC. A. Chorti has been partially supported by the EC through the Horizon Europe/JU SNS project ROBUST-6G (Grant Agreement no. 101139068), the IPAL Project CONNECTING, the TalCyb Chair in Cybersecurity and by the French government under the France 2030 ANR program “PEPR Networks of the Future” (ref. ANR-22-PEFT-0009). R. C. de Lamare has been supported by PUC-Rio, CNPq, FAPERJ and FAPESP.

\appendices
\section{Proof of Proposition \ref{l_updtae}}\label{appA}
\begin{proof}
Equation (\ref{eq_L}) can be written as   
\begin{equation}
\begin{aligned}
\mathbf{L}_u^{k+1}
& = \;  \arg\min_{\mathbf{L}_u}
\;\|\mathbf{L}_u\|_{*}
\;+\;\gamma_u\,\bigl\|\,\mathbf{L}_u - \beta_u\,\mathbf{L}_1\bigr\|_{F}^{2}
\;+\; \\
&\bigl\langle \mathbf{Y}^{k},\,\widehat{\mathbf{H}}_u - \mathbf{L}_u - \mathbf{S}_u^{k}\bigr\rangle
\;+\;\frac{\mu_{k}}{2}\,\bigl\|\,\widehat{\mathbf{H}}_u - \mathbf{L}_u - \mathbf{S}_u^{k}\bigr\|_{F}^{2}.
\end{aligned}
\end{equation}
Let us define
\(
\mathbf{R}_{L}^{\,k} 
:= \widehat{\mathbf{H}}_u - \mathbf{S}_u^{k}.
\) Then \(\widehat{\mathbf{H}}_u - \mathbf{L}_u - \mathbf{S}_u^{k} = \mathbf{R}_{L}^{\,k} - \mathbf{L}_u\). 
Then
\(
\bigl\langle \mathbf{Y}^{k},\,\mathbf{R}_{L}^{\,k} - \mathbf{L}_u\bigr\rangle
= \bigl\langle \mathbf{Y}^{k},\,\mathbf{R}_{L}^{\,k}\bigr\rangle 
  - \bigl\langle \mathbf{Y}^{k},\,\mathbf{L}_u\bigr\rangle,
\)
and\\
\(
\|\mathbf{R}_{L}^{\,k} - \mathbf{L}_u\|_{F}^{2}
= \|\mathbf{R}_{L}^{\,k}\|_{F}^{2}
  - 2\,\bigl\langle \mathbf{R}_{L}^{\,k},\,\mathbf{L}_u\bigr\rangle 
  + \|\mathbf{L}_u\|_{F}^{2}.
\)
When dropping all terms independent of \(\mathbf{L}_u\), the \(\mathbf{L}_u\)-dependent part is
\begin{equation}
\hspace{-0.1cm}
-\,\bigl\langle \mathbf{Y}^{k},\,\mathbf{L}_u\bigr\rangle
- \mu_{k}\,\bigl\langle \mathbf{R}_{L}^{\,k},\,\mathbf{L}_u\bigr\rangle
+ \frac{\mu_{k}}{2}\,\|\mathbf{L}_u\|_{F}^{2}
+ \gamma_u\,\|\mathbf{L}_u - \beta_u\,\mathbf{L}_1\|_{F}^{2}.
\end{equation}
By combining the two linear terms in \(\mathbf{L}_u\) we get:\\
$
-\,\bigl\langle \mathbf{Y}^{k},\,\mathbf{L}_u\bigr\rangle
- \mu_{k}\,\bigl\langle \mathbf{R}_{L}^{\,k},\,\mathbf{L}_u\bigr\rangle
= -\,\bigl\langle \mathbf{Y}^{k} + \mu_{k}\,\mathbf{R}_{L}^{\,k},\,\mathbf{L}_u\bigr\rangle.$
Then,
\begin{equation}
\begin{aligned}
& -\,\bigl\langle \mathbf{Y}^{k} + \mu_{k}\,\mathbf{R}_{L}^{\,k},\,\mathbf{L}_u\bigr\rangle
+ \frac{\mu_{k}}{2}\,\|\mathbf{L}_u\|_{F}^{2}
= \\
&\frac{\mu_{k}}{2}\,
  \Bigl\|\mathbf{L}_u - \frac{\mathbf{Y}^{k} + \mu_{k}\,\mathbf{R}_{L}^{\,k}}{\mu_{k}}\Bigr\|_{F}^{2} - 
\frac{1}{2\mu_{k}}\,
  \bigl\|\mathbf{Y}^{k} + \mu_{k}\,\mathbf{R}_{L}^{\,k}\bigr\|_{F}^{2}.
\end{aligned}
\end{equation}
Now let us define
\( 
\mathbf{B}^{\,k} 
:=\; \mathbf{R}_{L}^{\,k} + \frac{1}{\mu_{k}}\,\mathbf{Y}^{k}
 = \widehat{\mathbf{H}}_u - \mathbf{S}_u^{k} + \frac{1}{\mu_{k}}\,\mathbf{Y}^{k}.
\) 
By dropping the constant, the \(\mathbf{L}_u\)-subproblem becomes
\begin{equation} \label{eq24}
\min_{\mathbf{L}_u}
\|\mathbf{L}_u\|_{*}
\;+\;\gamma_u\,\|\mathbf{L}_u - \beta_u\,\mathbf{L}_1\|_{F}^{2}
\;+\;\frac{\mu_{k}}{2}\,\|\mathbf{L}_u - \mathbf{B}^{\,k}\|_{F}^{2}.
\end{equation}
To derive the sum of the two squared-Frobenius terms in (\ref{eq24}), 
let us set \(\alpha_{1} = \mu_{k}\) and \(\alpha_{2} = 2\,\gamma_u\), so \(\alpha_{1} + \alpha_{2} = \mu_{k} + 2\gamma_u\).  

\begin{equation}
\begin{aligned}
 \gamma_u\|\mathbf{L}_u - \beta_u\,\mathbf{L}_1\|_{F}^{2}
= & \frac{\alpha_{2}}{2}\Bigl(\|\mathbf{L}_u\|_{F}^{2}
                     - 2\,\langle \beta_u\mathbf{L}_1,\,\mathbf{L}_u\rangle \\
                    & + \|\beta_u\,\mathbf{L}_1\|_{F}^{2}\Bigr)
\end{aligned}
\end{equation}
and 
\begin{equation}
 \frac{\mu_{k}}{2}\,\|\mathbf{L}_u - \mathbf{B}^{\,k}\|_{F}^{2}
= \frac{\alpha_{1}}{2}\,\Bigl(\|\mathbf{L}_u\|_{F}^{2}
                       - 2\,\langle \mathbf{B}^{\,k},\,\mathbf{L}_u\rangle
                       + \|\mathbf{B}^{\,k}\|_{F}^{2}\Bigr).
\end{equation}

 Adding these yields
\begin{equation}
\frac{\alpha_{1} + \alpha_{2}}{2}\|\mathbf{L}_u\|_{F}^{2}
- \Bigl(\alpha_{2}\langle \beta_u\,\mathbf{L}_1,\,\mathbf{L}_u\rangle + \alpha_{1}\langle \mathbf{B}^{\,k},\,\mathbf{L}_u\rangle\Bigr)
+ C_1 
\label{eq27}
\end{equation}
with $C_1 = \frac{\alpha_{2}}{2}\,\|\beta_u\,\mathbf{L}_1\|_{F}^{2} + \frac{\alpha_{1}}{2}\,\|\mathbf{B}^{\,k}\|_{F}^{2}$.

\noindent We drop the constant $C_1$ and define
\(
\mathbf{G} 
:= \alpha_{1}\,\mathbf{B}^{\,k} + \alpha_{2}\,\beta_u\,\mathbf{L}_1,
\)
so (\ref{eq27}) becomes
\(
\frac{\alpha_{1} + \alpha_{2}}{2}\,\|\mathbf{L}_u\|_{F}^{2}
- \langle \mathbf{G},\,\mathbf{L}_u\rangle.
\) 
\begin{equation}
\frac{\alpha_{1} + \alpha_{2}}{2}\|\mathbf{L}_u\|_{F}^{2}
- \langle \mathbf{G},\,\mathbf{L}_u\rangle
 = \frac{\alpha_{1} + \alpha_{2}}{2}\Bigl\|\mathbf{L}_u - \frac{\mathbf{G}}{\alpha_{1} + \alpha_{2}}\Bigr\|_{F}^{2} + C_2 
\end{equation}
where the constant $C_2 = -\frac{1}{2(\alpha_{1} + \alpha_{2})}\,\|\mathbf{G}\|_{F}^{2}$.
We drop the constant $C_2$ and define
$\mathbf{T}^{\,k}
:= \dfrac{\mathbf{G}}{\alpha_{1} + \alpha_{2}}
= \dfrac{\mu_{k}\,\mathbf{B}^{\,k} + 2\,\gamma_u\,\beta_u\,\mathbf{L}_1}{\mu_{k} + 2\gamma_u}$.

\noindent Hence
\(
\gamma_u\,\|\mathbf{L}_u - \beta_u\,\mathbf{L}_1\|_{F}^{2}
+ \frac{\mu_{k}}{2}\,\|\mathbf{L}_u - \mathbf{B}^{\,k}\|_{F}^{2}
= \frac{\alpha_{1} + \alpha_{2}}{2}\,\|\mathbf{L}_u - \mathbf{T}^{\,k}\|_{F}^{2}
+ C_{12},
\)
where $C_{12} = C_1 + C_2$. The \(\mathbf{L}_u\)-subproblem becomes 
\begin{equation} \label{equa_32}
\mathbf{L}_u^{k+1}
= \arg\min_{\mathbf{L}_u}
\;\|\mathbf{L}_u\|_{*}
\;+\;\frac{\mu_{k} + 2\gamma_u}{2}\,\|\mathbf{L}_u - \mathbf{T}^{\,k}\|_{F}^{2}.
\end{equation}

\noindent We rewrite it in proximal form of the nuclear norm
\(
\min_{L} \Bigl\{ \| \mathbf{L}_u \|_{*} \;+\; \frac{1}{2\tau}\,\| \mathbf{L}_u - \mathbf{T}^k \|_{F}^{2} \Bigr\},
\)
where
\(
\tau \;=\; \frac{1}{\,\mu_{k} + 2\gamma_u\,}.
\)
Hence
\begin{equation}
\mathbf{L}_u^{\,k+1}
\;=\;
\mathrm{prox}_{\tau\,\|\cdot\|_{*}}\bigl(\mathbf{T}^{k}\bigr)
\end{equation}
and the closed-form solution of (\ref{equa_32}) is given by
\begin{equation}
\mathbf{L}_u^{k+1} =  \mathcal{D}_{1/(\mu_k+2\gamma_u)}(\mathbf{T}^k) =  \mathbf{U}\,\Bigl[\Sigma - \tfrac{1}{\mu_{k} + 2\gamma_u}\,\mathbf{I}\Bigr]_{+}\,\mathbf{V}^{\top},\\
\end{equation}
where $\mathcal{D}_{\tau}(\mathbf{T}^k)$ is the singular-value thresholding, that is $\mathbf{T}^k=\mathbf{U}\mathbf{\Sigma} \mathbf{V}^\top$, then $[\mathbf{\Sigma}-\tau \mathbf{I}]_+$ shrinks singular values as $\mathbf{L}_u^{k+1}= \mathbf{U} [\mathbf{\Sigma}-\tau\mathbf{I}]_+\mathbf{V}^T$. $\mathbf{I}$ is the identity matrix and \((x)_+ = \max\{x,0\}\) is applied element-wise.
\end{proof}
\section{Proof of Proposition \ref{s_updtae}} \label{appB}
\begin{proof}
Equation (\ref{eq_S}) can be written as
\begin{equation}
\begin{aligned}
 \mathbf{S}_u^{k+1}
= & \arg\min_{\mathbf{S}_u}
\;\lambda\,\|\mathbf{S}_u\|_{1}
\;+\;\bigl\langle \mathbf{Y}^{k},\,\widehat{\mathbf{H}}_u - \mathbf{L}_u^{k+1} - \mathbf{S}_u\bigr\rangle \\
& \;+\; 
\frac{\mu_{k}}{2}\,\bigl\|\widehat{\mathbf{H}}_u - \mathbf{L}_u^{k+1} - \mathbf{S}_u\bigr\|_{F}^{2}.
\end{aligned}
\end{equation}
For $
\mathbf{R}_{S}^{\,k} 
:= \widehat{\mathbf{H}}_u - \mathbf{L}_u^{k+1}$, 
$
\bigl\langle \mathbf{Y}^{k},\,\mathbf{R}_{S}^{\,k} - \mathbf{S}_u\bigr\rangle
= \bigl\langle \mathbf{Y}^{k},\,\mathbf{R}_{S}^{\,k}\bigr\rangle
  - \bigl\langle \mathbf{Y}^{k},\,\mathbf{S}_u\bigr\rangle,
  $
and
$
\|\mathbf{R}_{S}^{\,k} - \mathbf{S}_u\|_{F}^{2}
= \|\mathbf{R}_{S}^{\,k}\|_{F}^{2}
  - 2\,\bigl\langle \mathbf{R}_{S}^{\,k},\,\mathbf{S}_u\bigr\rangle
  + \|\mathbf{S}_u\|_{F}^{2}.
$
\begin{equation}
\hspace{-0.2cm}
\begin{aligned}
\langle \mathbf{Y}^{k},\,\mathbf{R}_{S}^{k} - &\mathbf{S}_u\rangle
+ \frac{\mu_{k}}{2}\|\mathbf{R}_{S}^{k} - \mathbf{S}_u\|_{F}^{2}
= - \bigl\langle \mathbf{Y}^{k} + \mu_{k}\,\mathbf{R}_{S}^{\,k},\,\mathbf{S}_u\bigr\rangle
\\ 
& + \frac{\mu_{k}}{2}\,\|\mathbf{S}_u\|_{F}^{2}
+ \bigl\langle \mathbf{Y}^{k},\,\mathbf{R}_{S}^{\,k}\bigr\rangle + \frac{\mu_{k}}{2}\,\|\mathbf{R}_{S}^{\,k}\|_{F}^{2} \text{,}
\end{aligned}
\end{equation}
\begin{equation}
\begin{aligned}
\implies  \bigl\langle \mathbf{Y}^{k},\,\mathbf{R}_{S}^{\,k} - & \mathbf{S}_u\bigr\rangle 
+ \frac{\mu_{k}}{2}\,\|\mathbf{R}_{S}^{\,k} - \mathbf{S}_u\|_{F}^{2} = \\
& \frac{\mu_{k}}{2}\,\Bigl\|\mathbf{S}_u - \Bigl(\mathbf{R}_{S}^{\,k} + \tfrac{1}{\mu_{k}}\,\mathbf{Y}^{k}\Bigr)\Bigr\|_{F}^{2} \;+\; C_3.
\end{aligned}
\end{equation}
Let us define
$
\mathbf{C}^{\,k}
:= \mathbf{R}_{S}^{\,k} + \frac{1}{\mu_{k}}\,\mathbf{Y}^{k}
= \widehat{\mathbf{H}}_u - \mathbf{L}_u^{k+1} + \frac{1}{\mu_{k}}\,\mathbf{Y}^{k}.$
By dropping the constant $C_3 = \bigl\langle \mathbf{Y}^{k},\,\mathbf{R}_{S}^{\,k}\bigr\rangle + \frac{\mu_{k}}{2}\,\|\mathbf{R}_{S}^{\,k}\|_{F}^{2}$, the \(\mathbf{S}_u\)-subproblem becomes
\begin{equation}
\mathbf{S}_u^{k+1}
= \arg\min_{\mathbf{S}_u}
\;\lambda\,\|\mathbf{S}_u\|_{1}
\;+\;\frac{\mu_{k}}{2}\,\|\mathbf{S}_u - \mathbf{C}^{\,k}\|_{F}^{2}.
\end{equation}
Since \(\|\mathbf{S}_u\|_{1} = \sum_{i,j} |S_{ij}|\) and \(\|\mathbf{S}_u - \mathbf{C}^{\,k}\|_{F}^{2} = \sum_{i,j}(S_{ij} - C^{\,k}_{ij})^{2}\), the problem decouples entrywise.  For each scalar \(s\) and \(c\):
\( 
\min_{s\in\mathbb{R}}
\;\lambda\,|s| \;+\;\frac{\mu_{k}}{2}\,(s - c)^{2}.
\) 
The subgradient condition
\(\;0 \in \lambda\,\partial|s| + \mu_{k}(s - c)\)
yields the soft‐threshold operator:
\begin{equation}
s^{*} 
= \mathcal{S}_{\lambda/\mu_{k}}(c)
= \mathrm{sign}(c)\,\max\Bigl(\lvert c\rvert - \frac{\lambda}{\mu_{k}},\,0\Bigr).
\end{equation}
Hence entrywise, we have
\begin{equation}
\mathbf{S}_u^{k+1} = \mathcal{S}_{\lambda/\mu_{k}}\bigl(\mathbf{C}^{\,k}\bigr). 
\end{equation}
This completes the proof.
\end{proof}

\section{Proof of Lemma \ref{lemLS}} \label{appC}
\begin{proof}
Let $\mathbf{P}^{k} = \widehat{\mathbf{H}}_u-\mathbf{L}_u^{k}-\mathbf{S}_u^{k}$.
$(\mathbf{L}_u^{k+1},\mathbf{S}_u^{k+1})$ minimizes
$\mathcal{L}_{\mu_k}(\cdot,\cdot,\mathbf{Y}^k)$, so
\begin{equation}
\begin{aligned}
 \mathcal{L}_{\mu_k}(\mathbf{L}_u^{k+1},\mathbf{S}_u^{k+1},\mathbf{Y}^k)
&\le \mathcal{L}_{\mu_k}(\mathbf{L}_u^{k},\mathbf{S}_u^{k},\mathbf{Y}^k)\\
= &
\mathcal{L}_{\mu_k}(\mathbf{L}_u^{k},\mathbf{S}_u^{k},\mathbf{Y}^{k-1})
+ \\
&\langle \mathbf{Y}^{k}-\mathbf{Y}^{k-1},\, \mathbf{P}^{k} \rangle\\
= &
\mathcal{L}_{\mu_k}(\mathbf{L}_u^{k},\mathbf{S}_u^{k},\mathbf{Y}^{k-1})
+ \mu_k \|\mathbf{P}^{k}\|_F^2\\
= & \, \mathcal{L}_{\mu_{k-1}}(\mathbf{L}_u^{k},\mathbf{S}_u^{k},\mathbf{Y}^{k-1})
+ \\
& \frac{\mu_k-\mu_{k-1}}{2}\,\|\mathbf{P}^{k}\|_F^2 + \mu_k \|\mathbf{P}^{k}\|_F^2
\end{aligned}
\end{equation}
\begin{equation}
\begin{aligned}
 \mathcal{L}_{\mu_k}(\mathbf{L}_u^{k+1},\mathbf{S}_u^{k+1},\mathbf{Y}^k)
&\le \mathcal{L}_{\mu_{k-1}}(\mathbf{L}_u^{k},\mathbf{S}_u^{k},\mathbf{Y}^{k-1}) \\
& + \Big(\mu_k + \frac{\mu_k-\mu_{k-1}}{2}\Big)\,\|\mathbf{P}^{k}\|_F^2
\end{aligned}
\label{eqLk}
\end{equation}
Repeating this back to $i = 1$ yields 
\begin{equation}\label{eqlk2}
\begin{aligned}
\mathcal{L}_{\mu_k}(\mathbf{L}_u^{k+1},\mathbf{S}_u^{k+1},\mathbf{Y}^{k})
& \le
\mathcal{L}_{\mu_{0}}(\mathbf{L}_u^{1},\mathbf{S}_u^{1},\mathbf{Y}^{0})\\
&+\sum_{i=1}^{k}\Big(\mu_i + \frac{\mu_i-\mu_{i-1}}{2}\Big)\,\|\mathbf{P}^{i}\|_F^2\\
\end{aligned}
\end{equation}
$\sum_{i=1}^{\infty} \mu_i \|\mathbf{P}^{i}\|_F^2 < \infty$, and with \( \{\mu_i\} \) non-decreasing and bounded, also \( \sum_i \|\mathbf{P}^{i}\|_F^2 < \infty \).
Therefore $\mathcal{L}_{\mu_k}(\mathbf{L}_u^{k+1},\mathbf{S}_u^{k+1},\mathbf{Y}^{k})$ is upper bounded according to
\begin{equation}
\begin{aligned}
\langle \mathbf{Y}^k,\mathbf{P}^k\rangle+\frac{\mu_k}{2}\|\mathbf{P}^k\|_F^2
&= \frac{\mu_k}{2}\Big\|\mathbf{P}^k+\frac{1}{\mu_k}\mathbf{Y}^k\Big\|_F^2
- \frac{1}{2\mu_k}\|\mathbf{Y}^k\|_F^2\\
&\ge -\frac{1}{2\mu_k}\|\mathbf{Y}^k\|_F^2.
\end{aligned}
\end{equation}
Hence, we have
\begin{equation}
 \mathcal{L}_{\mu_k}(\mathbf{L}_u^k,\mathbf{S}_u^k,\mathbf{Y}^k)
\ge f(\mathbf{L}_u^k)+g(\mathbf{S}_u^k)-\frac{1}{2\mu_k}\|\mathbf{Y}^k\|_F^2,
\end{equation}
\begin{equation}
\implies f(\mathbf{L}_u^k)+g(\mathbf{S}_u^k)
\le \mathcal{L}_{\mu_k}(\mathbf{L}_u^k,\mathbf{S}_u^k,\mathbf{Y}^k)
+ \frac{1}{2\mu_k}\|\mathbf{Y}^k\|_F^2.
\end{equation}
$\|\mathbf{L}_u^k\|_* \le f(\mathbf{L}_u^k)$ and 
$\lambda\|\mathbf{S}_u^k\|_1 \le g(\mathbf{S}_u^k)$, then 
\begin{equation}
\begin{aligned}
\|\mathbf{L}_u^k\|_* + \lambda\|\mathbf{S}_u^k\|_1 
&\le f(\mathbf{L}_u^k)+g(\mathbf{S}_u^k) \\
&\le \mathcal{L}_{\mu_k}(\mathbf{L}_u^k,\mathbf{S}_u^k,\mathbf{Y}^k)
+ \frac{1}{2\mu_k}\|\mathbf{Y}^k\|_F^2
\end{aligned}
\end{equation}
is upper bounded. Therefore, $\{\mathbf{L}_u^k\}$ and $\{ \mathbf{S}_u^k\}$ are upper bounded.
\end{proof}

\bibliographystyle{IEEEtran}
\bibliography{references}
\end{document}